%% file: clean.tex
\pdfoutput=1

\documentclass[11pt]{scrartcl}
\usepackage[a4paper, total={16cm, 24cm}]{geometry}
\usepackage{natbib}
\renewcommand\tableofcontents{\listoftoc*{toc}} 
 
\usepackage{authblk}

\author[1]{Niclas Boehmer}
\author[2]{Piotr Faliszewski}
\author[2]{Łukasz Janeczko}
\author[3]{\authorcr Dominik Peters} 
\author[4]{Grzegorz Pierczyński}
\author[5]{Šimon Schierreich}
\author[4]{\authorcr Piotr Skowron}
\author[2]{Stanisław Szufa}
\affil[1]{Harvard University} 
\affil[2]{AGH University, Kraków} 
\affil[3]{CNRS, LAMSADE, Universit\'e Paris Dauphine - PSL}
\affil[4]{University of Warsaw}
\affil[5]{Czech Technical University in Prague}

\usepackage[pagebackref]{hyperref}
\hypersetup{
	pdfencoding=auto, 
	psdextra,
	colorlinks=true,
	citecolor=green!50!black,
	linkcolor=red!60!black,
}

\usepackage{natbib}  
\usepackage{caption} 
\usepackage{algorithm}
\usepackage{algorithmic}

\usepackage{graphicx} 
\usepackage{microtype}
\usepackage{amsmath}
\usepackage{amssymb}
\usepackage{subcaption}
\usepackage{nicefrac}
\usepackage[textsize=tiny]{todonotes}
\usepackage{amsthm}
\usepackage{booktabs}
\usepackage{paralist}
\usepackage{thm-restate}

\usepackage[pagebackref]{hyperref}
\hypersetup{
		pdfencoding=auto, 
		psdextra,
		colorlinks=true,
		citecolor=green!40!black,
		linkcolor=red!50!black,
		urlcolor=blue!80!black
	}

\usepackage[nameinlink]{cleveref}
\usepackage{newfloat}
\usepackage{listings}
\DeclareCaptionStyle{ruled}{labelfont=normalfont,labelsep=colon,strut=off} 
\floatstyle{ruled}
\newfloat{listing}{tb}{lst}{}
\floatname{listing}{Listing}

\title{Evaluation of Project Performance in Participatory Budgeting}

\newcommand{\calS}{\mathcal{S}}

\usepackage{pgfplots}

\newcommand{\figurecut}[1]{}
\newcommand{\cutfornow}[1]{}
\newcommand{\np}{\ensuremath{{\mathrm{NP}}}}
\newcommand{\conp}{\ensuremath{{\mathrm{coNP}}}}

\newcommand{\sharpp}{\ensuremath{{\mathrm{\#P}}}}

\newcommand{\fpt}{{{\mathrm{FPT}}}}

\newcommand{\av}{\textsc{AV}}

\newcommand{\phragmen}{\textsc{Phragm{\'e}n}}
\newcommand{\eqphragmen}{\textsc{Eq/Phragm{\'e}n}}

\newcommand{\equalshares}{\textsc{Equal-Shares}}

\newcommand{\greedyav}{\textsc{greedy}\textsc{AV}}

\newcommand{\ph}{\textsc{Ph}}
\newcommand{\eq}{\textsc{Eq}}

\newcommand{\costred}{{{\mathrm{cost\hbox{-}red}}}}
\newcommand{\optimistadd}{\mathrm{optimist\hbox{-}add}}
\newcommand{\pessimistadd}{{{\mathrm{pessimist\hbox{-}add}}}}
\newcommand{\randomadd}{{{\mathrm{50\%\hbox{-}add}}}}
\newcommand{\singletonadd}{{{\mathrm{singleton\hbox{-}add}}}}
\newcommand{\rivalred}{{{\mathrm{rival\hbox{-}red}}}}

\newcommand{\setCover}{\textsc{Set-Cover}}

\newcommand{\rxthreec}{\textsc{RX3C}}

\newcommand{\naturals}{\mathbb{N}}

\newtheorem{theorem}{Theorem}[section]
\newtheorem{definition}[theorem]{Definition}

\DeclareMathOperator{\cost}{cost}

\newcommand{\ssnote}[1]{}
\newcommand{\pfnote}[1]{}
\newcommand{\dpnote}[1]{}

\title{Evaluation of Project Performance in Participatory Budgeting}
\date{\vspace{-1cm}}
	
\begin{document}

\maketitle

\begin{abstract}
	\begin{center}
		\textbf{\textsf{Abstract}} \smallskip
	\end{center}
  We study ways of evaluating the performance of losing projects in
  participatory budgeting (PB) elections by seeking actions that would
  have led to their victory.  We focus on lowering the projects'
  costs, obtaining additional approvals for them, and asking
  supporters to refrain from approving other projects: The larger a
  change is needed, the less successful is the given project. We seek
  efficient algorithms for computing our measures and we analyze and
  compare them experimentally. We focus on the $\greedyav$,
  $\phragmen$, and $\equalshares$ PB rules.
\end{abstract}

\vspace{13pt}
\hrule
\vspace{8pt}
{\small\tableofcontents}
\vspace{20pt}
\hrule
\newpage

\section{Introduction}\label{sec:intro}

The idea of participatory budgeting (PB) is to let members of a local
community---such as a city or its district---decide how a certain
budget $B$ should be
spent~\citep{cab:j:participatory-budgeting,goe-kri-sak-ait:c:knapsack-voting,rey-mal:t:pb-survey}. To
this end, some members of the community first submit projects,
together with the costs of implementing them, and then the whole
community votes on which of them should be funded, typically by
indicating which projects they approve (but see, e.g., the work of
\citet{fai-ben-gal:c:pb-formats} for an analysis of other ballot formats). The funded projects
are selected using the $\greedyav$ rule, which looks at the projects
one by one, starting with those that received the most approvals,
and selects a project if its cost is within the
still-available budget. For example, consider a PB instance with
projects $a$, $b$, $c$,~$d$,~$e$, 10 voters, and budget $B = 10$ (we
consider a fairly small city), where the project costs and the votes
are as follows:
\begin{center}
  \begin{tabular}{cc|cccccccccc}
    \multicolumn{2}{c}{} & \multicolumn{10}{c}{voters} \\
    proj. & cost & $x_1$ & $x_2$ & $x_3$ & $y_1$ & $y_2$ & $y_3$ & $z_1$ & $z_2$ & $z_3$ & $z_4$ \\
    \midrule
    $a$ & 7   & \checkmark & \checkmark & \checkmark & \checkmark & \checkmark & \checkmark & - & - & - & - \\                            
    $b$ & 4   & - & - & - & - & - & - & \checkmark & \checkmark & \checkmark & \checkmark \\                                 
    $c$ & 3   & - & - & - & \checkmark & \checkmark & \checkmark & - & - & - & - \\                                 
    $d$ & 2   & - & \checkmark & \checkmark & - & - & - & - & - & - & - \\                                 
    $e$ & 2   & \checkmark & - & - & - & - & - & - & - & - & - \\       
    
    \bottomrule
  \end{tabular}
\end{center}\medskip

\noindent In particular, project $a$ costs $7$ units and is approved
by 6 voters ($x_1, x_2, x_3$ and $y_1, y_2, y_3$). $\greedyav$ first
considers project $a$ (as it is approved by the largest number of
voters) and selects it.  Then it looks at project $b$ (approved by $4$
voters, costing $4$ units), but it does not select it as there are
only~$3$ units of budget left.  Next it considers project $c$ and
selects it, using up the whole remaining budget.

When the city announces the results, it may wish to present additional information to explain why projects lost or won.
In particular, proposers and supporters of losing projects may wish to know how close their project was to winning.
A city could do this by publishing an ``\emph{information package}'' for each losing project.
In the example, for the losing project $b$, the package could note that the project would have won had
it been cheaper by one unit, or had it
been supported by two or three additional voters (depending on
how ties are broken).
Having such an information package would 
make the whole
process more transparent.

The $\greedyav$ rule has the disadvantage that it frequently underserves minorities (indeed, in our
example voters $z_1, \ldots, z_4$ do not approve any winning projects, even though
they form $40\%$ of the electorate).
Researchers have proposed new 
PB rules, such as
$\equalshares$~\citep{pet-sko:c:welfarism-mes,pet-pie-sko:c:pb-mes}
and
$\phragmen$~\citep{bri-fre-jan-lac:c:phragmen,los-chr-gro:c:phragmen-pb},
that aim at proportional representation of the voters and avoid this problem.
These rules are more involved than $\greedyav$, making it important to explain the outcome to participants. We study performance measures like the ones mentioned above for PBs using these rules, to help produce good information packages.
This will be particularly useful for $\equalshares$, which was recently used
in real-life PB elections in Wieliczka (Poland) and Aarau
(Switzerland).

While the measure based on reducing projects' costs applies directly to proportional rules,
there are some challenges in adapting the measure of how many more approvals are necessary for funding.
The reason is
that under $\greedyav$ it only matters \emph{how many} approvals we
add, but under proportional rules it is also important \emph{where} we
add them. We illustrate this issue using $\equalshares$, which roughly
works as follows:
Given $n$~voters and budget~$B$, the rule assigns $\nicefrac{B}{n}$
units of budget to each voter. Then, it proceeds in rounds, where in
each round
a group of voters \emph{buys} a project if (1) they jointly have
sufficient funds, and (2) they are the largest group that can afford a
project at this point of time; the voters from this group share the
cost of the project equally among themselves.\footnote{If there are voters who do not have enough funds left to
  pay their full share, they pay the fraction that they have left, but
  they also count as a fraction of a voter toward the group's size.}  In our example,
each voter obtains one unit of
budget, and in the first round 
voters $z_1, z_2, z_3, z_4$ buy project $b$ (they are the largest
group that can afford a project; project $a$ has more approvals, but
its supporters do not have sufficient funds). Then, voters $y_1$,
$y_2$, and $y_3$ buy project $c$ and, finally, voters $x_2$ and $x_3$
buy project $d$. The final outcome is $\{b,c,d\}$, under which only
voter $x_1$ is left without an approved winning project.
Now let us see how this outcome would change if $e$ got some
additional approvals.  If it got approvals from $x_2$ and $x_3$, then
$\equalshares$ would certainly select it instead of $d$.
Yet, if $e$ got two additional approvals from $z_1$ and $z_2$ then
it would still lose; these voters would spend all their funds on $b$
and would not 
help $e$.

We address this issue with the adding-approvals measure
by considering (a)~the smallest number of approval additions needed
for a project's victory (for the right selection of voters
who add the approvals), (b)~the smallest number of approval additions
that suffice for a victory no matter which voters perform them, as
well as (c)~a variant based on randomization and (d)~a variant based
on adding new voters.

Finally, under a proportional rule, 
a project may also lose because
it is supported by voters for whom the rule chooses other projects
that they approve. Indeed, this is exactly why $\equalshares$ would
not select project~$e$ even if it were additionally approved by
$z_1$ and $z_2$. If these voters did not approve $b$, their support
for $e$ would suffice for its victory. Knowing that a project could
have won if some of its supporters refrained from voting for other
projects is useful both for voters to understand how the rule works, and for project proposers for planning their election campaigns.

\paragraph{Technical Contribution.}
We are interested in two aspects of our measures. First, as real-life
PB instances can include hundreds of projects and tens of
thousands of voters, we seek fast algorithms for computing them in practice.
As some of our measures are NP-hard to compute, we find $\fpt$
algorithms for them that work well in realistic settings.

Second, we compute our measures for projects from real-life PB
instances from Pabulib~\citep{fal-fli-pet-pie-sko-sto-szu-tal:c:pabulib}.  We find that
measures of the same type are strongly correlated, but that measures of different types
may disagree about how well a project performed. All measures
provide helpful insights, which we demonstrate
by constructing an example information package based on the Pabulib vote data from the
2023 Wieliczka ``Green Million'' PB, where $\equalshares$ was used for the first time.

Proofs of all statements and much additional information on our experiments can be found in the appendix.

\section{Preliminaries}\label{sec:prelim}

A \emph{participatory budgeting} (PB) instance $E = (P,V,B)$ consists
of a set of projects $P = \{p_1, \ldots, p_m\}$, a set of voters
$V = \{v_1, \ldots, v_n\}$, and budget $B \in \naturals$.  Each voter
$v$ has a set $A(v) \subseteq P$ of projects that he or she approves
(referred to as his or her \emph{approval set}), and each project
$p \in P$ has a price for implementing it, denoted $\cost(p)$.
We extend this notation 
so that for a project $p$, $A(p)$ is the set of voters who approve
it, and we refer to $|A(p)|$ as the approval score of project~$p$.
A set
$S \subseteq P$ of projects is feasible if its cost, denoted as
$\cost(S) = \sum_{p \in S}\cost(p)$, is at most~$B$.  
A PB rule is a function that given a PB instance outputs a feasible
set $S$ of selected projects (i.e., our rules are resolute and, so,
give unique outcomes). We refer to projects in $S$ as the
\emph{selected} (or, \emph{funded}) ones, and to the remaining ones
as \emph{losing}.

We consider three PB rules, $\greedyav$, $\phragmen$
\citep{bri-fre-jan-lac:c:phragmen,los-chr-gro:c:phragmen-pb}, and
$\equalshares$~\citep{pet-sko:c:welfarism-mes,pet-pie-sko:c:pb-mes}. Each
of them starts with an empty set of projects $W$, performs a sequence
of rounds, where in each round it extends $W$ with a single project,
and eventually outputs $W$ as the final outcome.  Whenever they
encounter an internal tie (i.e., two or more projects that fulfill a
given condition), they break it using a given, prespecified order over
the projects. 
Our rules work as follows, letting $E = (P,V,B)$ be the
input PB instance:

\begin{description}
\item[$\greedyav$ (\av).] In each round, $\greedyav$ considers pro\-ject~$p$
  with the highest approval score, that it has not considered yet.  If
  $\cost(W)+\cost(p) \leq B$ (i.e., if it can afford to fund $p$) then
  it includes $p$ in $W$. The rule terminates upon considering all the
  projects.

\item[$\phragmen$ (\ph).]

  The voters start with empty virtual bank accounts, but receive funds
  in a
  continuous manner, one unit of funds per one unit of time. As soon as there is a project
  $p \in P \setminus W$ such that the voters in $A(p)$ have $\cost(p)$
  funds in total and $\cost(W)+\cost(p) \leq B$, the rule includes $p$
  in $W$ and resets the bank accounts of the voters from $A(p)$ to
  zero (these voters \emph{buy} the project). The process
  stops when for every
  project $p$ with at least one approval it holds that
  $\cost(W) + \cost(p) > B$.

\item[$\equalshares$ (\eq).] This rule also uses voters' virtual bank
  accounts, but it initiates them to $B/|V|$ per voter and does not
  provide further funds. Each round proceeds as follows, where $b_i$
  is the current account balance of voter $v_i \in V$.
  The idea is to select a project that its supporters can afford, and
  such that each of them has to cover as small a fraction of its cost
  as possible.  Formally, a project $p \in P \setminus W$ is
  affordable if there is $q(p) \in [0,1]$ such that:
  \[
    \textstyle
    \sum_{v_i \in A(p)} \min\big( b_i, q(p) \cdot \cost(p)\big) = \cost(p).
  \]
  For each voter $v_i \in A(p)$, we let $q_i(p)$ be the fraction of
  $\cost(p)$ that $v_i$ needs to cover; it is $q(p)$ if
  $b_i \geq q(p) \cdot \cost(p)$ (i.e., if $v_i$ can afford its full
  share) and it is $\nicefrac{b_i}{\cost(p)}$ otherwise.  The rule
  selects an affordable project~$p$ with the smallest value of $q(p)$,
  includes it in $W$, and charges each $v_i \in A(p)$ with
  $q_i(p) \cdot \cost(p)$ (if there are no affordable projects, then
  the rule terminates).\footnote{In the language of \citet{bri-for-lac-mal-pet:c:prop-approval-pb}, this definition is based on \textit{cost utilities}. The definition from the introduction matches the formal definition if we define the weighted size of the
    voter group that approves~$p$ as
    $\sum_{v_i \in A(p)} \nicefrac{q_i(p)}{q(p)}$. Each voter that
    contributed the full share counts as~$1$ in this sum, but those
    who contributed less count as respective fractions. One can verify
    that $\equalshares$ selects a project approved by a group
    with the largest weighted size.}
\end{description}
A well-known issue of $\equalshares$ is that it is not
exhaustive, i.e., upon termination there still may be sufficient
budget left to fund more projects. Hence, in practice one needs to apply one of several \emph{completion methods}. In most
of our experiments we do the following: When $\equalshares$
terminates, we start the $\phragmen$ rule, but with voters' bank
accounts initiated to their value at the end of
$\equalshares$. We call this rule
$\eqphragmen$.
In Wieliczka and Aarau, a different completion method was used, which, unfortunately, is too
computationally intensive for our full set of experiments (see \Cref{app:prelim}).

Each of the above-described PB rules can be computed in polynomial
time using 
a \emph{round-based} algorithm, which executes each round following
the definition.
Many of our measures can be computed while running a round-based algorithm, by performing some additional steps in
each round.
\begin{definition}
  Let $f$ be a PB rule. Let $\mathrm{measure}_E(p)$ be
  a function that takes as input a PB instance $E$ and project~$p$,
  and let $t \colon \naturals \rightarrow \naturals$ be some
  function. We say that $\mathrm{measure}_E(p)$ can be computed
  \emph{alongside $f(E)$ at a cost of $t(|E|)$ per round} if it is
  possible to compute its value using a round-based algorithm for
  $f(E)$, extended to perform at most $t(|E|)$ additional
  computational steps in each round.
\end{definition}

\section{The Measures and How to Compute Them}\label{sec:measures}

\begin{table}
	\centering
  \begin{tabular}{r|ccc}
    \toprule
    measure & $\av$ & $\ph$ & $\eq$ \\
    \midrule
    $\costred$      & \small along/$O(1)$ & \small along/$O(1)$        & \small along/$O(1)$ \\
    $\optimistadd$  & \small along/$O(1)$ & \small along/$O(n \log n)$ & \small along/$O(n \log n)$ \\
    $\randomadd$    & \small along/$O(1)$ & {\small sampling}          & {\small sampling} \\
    $\pessimistadd$ & \small along/$O(1)$ & \small $\np$-com./$\fpt$   & \small $\fpt$ \\
    $\singletonadd$ & \small along/$O(1)$ &  \small along/$O(1)$       & \small brute-force \\
    $\rivalred$     & {\small sampling}   & {\small sampling}          & {\small sampling}  \\    
    \bottomrule
  \end{tabular}
\caption{\label{tab:algorithms}Summary of our algorithms for computing
  the performance measures. By along/$O(1)$ and along/$O(n \log n)$ we
  mean that the measure can be computed alongside the rule, 
  with $O(1)$ or $O(n \log n)$ additional cost per round (where $n$ is
  the number of voters). By sampling, we mean algorithms based on
  simulating a given action a number of times. By $\fpt$, we mean the algorithm from \Cref{thm:fpt}. By
  brute-force, we mean adding singleton voters one by one.}
\end{table}

In this section, we describe our measures and provide ways of
computing them. While we also analyze their worst-case computational
complexity, our focus is on obtaining practically usable algorithms
that can be applied to PB elections from
Pabulib~\citep{fal-fli-pet-pie-sko-sto-szu-tal:c:pabulib}.
We focus on $\greedyav$ ($\av$), $\phragmen$ ($\ph$), and
$\equalshares$ ($\eq$) as this allows for a clean presentation.
By combining algorithms for $\eq$ and $\ph$, one can obtain algorithms that work for $\eqphragmen$.
We summarize our algorithms in \Cref{tab:algorithms}.

The common feature of our measures is that they correspond to specific
actions that either the voters or the project proposers could have
taken in the election.  In this respect, they are closely related to
the
margin-of-victory~\citep{mag-riv-she-wag:c:stv-bribery,car:c:stv-margin-of-victory,xia:margin-of-victory}
and, more broadly, bribery
notions~\citep{fal-hem-hem:j:bribery,fal-rot:b:control-bribery,yan:c:approval-multiwinner-destructive-bribery}.  In
particular, we borrow ideas from the work of
\citet{fal-sko-tal:c:bribery-multiwinner-success} about bribery in multiwinner elections.

Except for rivalry reduction (see \Cref{sec:rivalry-red}), our
measures are well-defined in all but a few pathological cases (such as
a project without any approvals). We list these issues
in~\Cref{sec:undefined}.

\subsection{Cost-Reduction Measure}

Our conceptually simplest measure is the one based on reducing the
project's cost. The measure was also studied by \citet{bau-boe-hil:c:pb-manip}. Formally, we define it as follows.

\begin{definition}
	Let $f$ be a PB rule, let $E = (P,V,B)$ be a PB instance, and let
	$p \in P \setminus f(E)$ be a losing project. We define the
	cost-reduction measure of $p$ in $E$, denoted $\smash{\costred^f_E(p)}$, to
	be the largest value such that if we replace~$p$'s cost with it,
	then $f$ selects~$p$. 
\end{definition}

\noindent
That is, $\costred_E(p)$ is the project's cost after the smallest
possible reduction that gets $p$ funded (we drop the superscript
denoting the rule when it is clear from the context).

For $\av$, $\ph$, and $\eq$, it is immediate that we can compute the cost-reduction measure in
polynomial time using binary search, but it would require recomputing
the rules multiple times. 
Instead, we compute it alongside our rules, for each round
finding the largest cost at which our project can be selected right
then, and outputting the largest of these
values. 
A similar algorithm is used by \citet{bau-boe-hil:c:pb-manip} for $\av$.

\begin{restatable}{proposition}{propCostRedPoly}\label{prop:costRed:poly}
  For $\av$, $\ph$, and $\eq$,
  $\costred_E(p)$ can be computed alongside the rule, at an $O(1)$
  cost per round.
\end{restatable}

\subsection{Add-Approvals Measures}

Next we consider the number of additional approvals needed by a losing
project to be funded.  Our idea is to consider different attitudes
toward risk when asking random voters to add approvals for a given
project: An optimist would hope to be lucky and obtain new approvals
exactly from the right voters, a pessimist would prefer to ensure that
the project is funded irrespective of which voters add the approvals, and
a middle-ground position is to require some fixed probability of
success.  Formally, we express this idea as follows.

\begin{definition}
  Let $f$ be a PB rule, let $E = (P,V,B)$ be a PB instance,
  and let $p \in P \setminus f(E)$ be a losing project. Then:
  \begin{enumerate}
  \item $\optimistadd^f_E(p)$ is the smallest number $\ell$ such
    that it is possible to ensure that $p$ is funded by choosing
    $\ell$ voters and extending their approval sets with $p$.
  \item $\pessimistadd^f_E(p)$ is the smallest number $\ell$ such
    that for each subset of $\ell$ voters who do not approve $p$,
    extending their approval sets with $p$ ensures that $p$ is funded.
  \item $\randomadd^f_E(p)$ is the smallest number $\ell$ such that
    if $\ell$ voters selected uniformly at random (among those who
    originally do not approve~$p$) extend their approval sets with
    $p$, then $p$ is funded with probability at least $50\%$.
  \end{enumerate}
\end{definition}

The optimist measure  
was previously considered by
\citet{fal-sko-tal:c:bribery-multiwinner-success} in multiwinner
voting, whereas the $50\%$-threshold one is inspired by an analogous
notion used
by~\citet{boe-bre-fal-nie:c:counting-bribery,boe-bre-fal-nie:c:robustness-single-winner}
in the single-winner setting, and by
\citet{boe-fal-jan-kac:c:pb-robustness} in the robustness analysis of
PB outcomes.
For each rule~$f$, each PB instance~$E$, and each losing project~$p$
we have the following:
\begin{align*}
  \optimistadd_E(p) \leq \randomadd_E(p) \leq \pessimistadd_E(p).
\end{align*}
For $\greedyav$ all three measures are equal and immediate to compute,
so we mostly focus on the proportional rules. Computing the optimist
measure is easy as it suffices to consider each round independently
and find the ``richest'' voters whose approval for $p$ would lead to
funding~it.

\begin{restatable}{proposition}{proOptimistAddPoly}\label{prop:optimitstadd:poly}
  For $\ph$ and $\eq$, 
  $\optimistadd_E(p)$ can be computed alongside the rule, at an
  $O(n\log n)$ cost per round, where $n$ is the number of voters.  For
  $\av$ it can be computed alongside the rule at an $O(1)$ cost per
  round.
\end{restatable}

\noindent
On the negative side, deciding if the pessimist measure has at least a
given value under 
$\phragmen$
is $\conp$-complete (the result
for $\equalshares$ remains elusive, although we also suspect
computational hardness).
Our proof looks at the complement of our problem, where we ask if it
is possible to add a certain number of approvals for $p$ without
getting it funded, and we give a reduction from 
a variant of 
\textsc{Set Cover} where we ask for an exact cover (i.e., a family of
pairwise disjoint sets that union up to a given universe).
The idea is to set up a (somewhat intricate) PB instance where
$\phragmen$ interleaves between selecting candidates from
the to-be-covered universe
and certain dummy candidates, whose selection resets voters' budgets
to a fixed state. If we add approvals for a designated candidate~$p$
to voters that form an exact cover, then this process goes as planned, but if
we add approvals to two voters whose corresponding sets overlap, then
$p$ gets funded. 

\begin{restatable}{theorem}{thmphragmenpessimist}\label{thm:phragmen-pessimist}
  For $\ph$, the problem of deciding if $\pessimistadd_E(p)$ is at
  least a given value
  $\ell$ 
  is $\conp$-complete, even if all projects have unit cost.
\end{restatable}

Fortunately, for $\phragmen$ and $\equalshares$ we can compute the
pessimist measure using an $\fpt$ algorithm parameterized by the
number of originally funded projects.  The key insight is to group 
voters by their bank account balances at the start of each
round. Unlike many $\fpt$ algorithms, this one is indeed
practical (since in practice the number of groups is small) and we use it in our experiments (using Gurobi).

\begin{restatable}{theorem}{thmFPT}\label{thm:fpt}   
  For $\ph$ and $\eq$, there is an algorithm that computes
  $\pessimistadd_E(p)$ and runs in $\fpt$ time with respect to
  parameter $|f(E)|$, i.e., the number of rounds.
\end{restatable}
\begin{proof}[Proof sketch (\phragmen{})]
  Let $E = (P,V,B)$ be a PB instance with losing project $p$ (for ease
  of exposition, we assume $p$ to be last in the tie-breaking
  order). We will show how to compute the largest number of approvals
  whose addition \emph{does not} lead to funding $p$; 
  $\pessimistadd_E(p)$ is one larger.  Let $k = |\phragmen(E)|$ be the
  number of funded projects or, equivalently, the number of rounds
  performed by the rule. We assume that $p$ is approved by at least
  one voter in $V$.

  For each round $i$, let $m_i$ be the difference between $\cost(p)$
  and the total funds that approvers of $p$ have in round $i$.  Our goal is to find a largest group of voters who
  do not approve $p$ and whose total funds in each round
  $i$ are at most $m_i$.
  
  For each voter $v_i \in V \setminus A(p)$, we define its
  \emph{balance vector} $(b^i_1,\ldots, b^i_k)$, which contains the
  balance of the voter's bank account right before each round. We
  partition the voters not approving $p$ into \emph{voter-types}
  $T=\{T_1,\ldots,T_t\}$, where each type consists of voters with
  identical balance vectors.

  To solve the problem, we form an integer linear program (ILP).
  For every voter-type $T_i$ we have a nonnegative integer
  variable $x_{T_i}$ that represents the number of voters of this type
  that will additionally approve~$p$. For each round $i$ we form the
  following \emph{round-constraint}: 
  \[
    \textstyle \sum_{j\in\{1, \ldots, t\}} x_{T_j}\cdot b^{T_j}_i \leq m_i.
  \]
  The objective function is to 
  maximize the sum of all the $x_{T_i}$ variables. Our algorithm outputs
  the value of this sum plus $1$.

  Finally, we note that there are at most $O(2^k)$ voter types
  (indeed, each voter type corresponds to a $k$-dimensional $0$/$1$
  vector, which has $1$ in position $i$ if a voter approves---and,
  hence, pays for---the candidate selected in round $i$).  Thus,
  the number of variables in our ILP is $O(2^k)$
  ; we can solve it using the classic algorithm of
  \citet{len:j:integer-fixed} in $\fpt$ time with respect to $k$.
  The algorithm can be
  tweaked to also work for \equalshares{}.
\end{proof}

For the $50\%$-threshold measure, we resort to sampling. That is,
given rule $f$, a PB instance $E$ and a losing project~$p$, we iterate
over numbers $\ell$ of approvals to add and for each of them we repeat
the following experiment $t$ times (where $t$ is a parameter): We add
approvals for $p$ to $\ell$ voters chosen uniformly at random (among
those not approving $p$) and we compute $f$ on the thus-modified
instance. We terminate for the smallest value $\ell$ where $p$ was
funded at least $t/2$ times.
We use sampling because the construction from  Theorem~\ref{thm:phragmen-pessimist} also shows that for
$\phragmen$, the problem of evaluating the probability that a given
project is funded after randomly adding a given number of approvals is
$\sharpp$-complete.

\subsection{Add-Singletons Measure}\label{add:single}
Instead of asking existing voters to approve some project~$p$, one
can also recruit additional voters, who would only approve $p$. This
gives rise to the following measure.

\begin{definition}
  Let $f$ be a PB rule, let $E = (P,V,B)$ be a PB instance,
  and let $p \in P \setminus f(E)$ be a losing project. Then
  $\smash{\singletonadd^f_E}(p)$ is the smallest number $\ell$ such that
  if we extend $V$ with $\ell$ voters who only approve $p$, then $f$
  selects~$p$.
\end{definition}
For $\av$, this is equal to the measures from the previous section.
For $\phragmen$, $\singletonadd_E(p)$ is upper-bounded by
$\optimistadd_E(p)$: In each round the newly added voters always have
at least as much money as the original ones (until $p$ is
selected). Under $\equalshares$ adding new voters changes the initial
balances of the voters, which can change the overall execution of the
rule and, hence, there is no clear relation between the two
measures. In fact, under $\equalshares$ it is even possible that a
project is funded after adding $\ell$ voters, but may fail to be
funded after adding $\ell+1$ of them~\citep[Proposition
A.3]{lac-sko:b:approval-survey}; see also the experiments in
\Cref{be-co}.

Computing the value of $\singletonadd_E(p)$ is easy for $\phragmen$
because we can compute how much additional funds each new voter
would bring in each round.  For $\equalshares$, due to its
nonmonotonic behavior, we resort to a brute-force approach: We keep
adding voters one by one and recompute the rule each time (in our
experiments, we were forced to add larger groups to speed up
computation).

\begin{restatable}{proposition}{propRandomAddPoly}\label{prop:randomAdd:poly}
  For $\av$ and $\ph$, $\singletonadd_E(p)$ can be computed
  alongside the rule, at an $O(1)$ cost per round.
\end{restatable}

\subsection{Rivalry-Reduction Measure}\label{sec:rivalry-red}
Under proportional rules, a project may lose because its supporters
also approve other projects, on which they spend their virtual money.
Thus, another strategy that a project proposer could employ to increase
a project's chances of success is to try to convince its supporters to
not approve other projects.

\begin{definition}
  Let $f$ be a PB rule, let $E = (P,V,B)$ be a PB instance, and let
  $p \in P \setminus f(E)$ be a losing project. Then $\smash{\rivalred^f_E}(p)$ is
  the smallest number $\ell$ such that if we select $\ell$ voters
  uniformly at random (among those who approve~$p$)
  and change them to only approve $p$, then $p$ is funded with probability at least $50\%$.
\end{definition}
\noindent
This measure is not always defined: A project with too
few voters will not win even if they do not support
any other projects.
For the sake of focus, we do not study the optimist and pessimist variants,
except to note that even the optimist variant would be
$\np$-complete to compute (by adapting proofs
on bribery and control in single-winner approval
voting~\citep{fal-hem-hem:j:bribery,hem-hem-rot:j:destructive-control}).
\begin{restatable}{proposition}{rivalredproof}
   For $\av$, $\ph$, and $\eq$, the problem of
  deciding, given a PB instance $E$, a losing project $p$, and an
  integer $\ell$, if it is possible to ensure $p$'s victory by
  changing at most $\ell$ votes that originally approve $p$ to only
  approve $p$ is $\np$-complete, even when $B = 1$ and every project costs 1.
\end{restatable}
\noindent To compute $\rivalred$, we use an analogous sampling
approach as in the case of $\randomadd$.

One may also wonder why we consider \emph{all} the voters who approve
$p$ and not only those who approve $p$ \emph{and} some further
project(s). We chose this approach to capture a campaign that reaches
the supporters of the project randomly.

\section{The Measures in Practice}\label{sec:pract}
Next, we analyze the behavior and correlation of our measures on
real-world PB instances from Pabulib \cite{fal-fli-pet-pie-sko-sto-szu-tal:c:pabulib},
focusing on $\eqphragmen$ here. 
Afterwards, we present a detailed study of the PB election from
Wieliczka. 
We include further details, as well as
results for $\phragmen$ and $\greedyav$, in \Cref{app:pract}.

\newcommand{\corr}[1]{\tikz{\pgfmathsetmacro{\myperc}{#1*#1*33}\node [transform shape, rounded corners=1pt,fill=blue!\myperc,inner sep=0, minimum width=22pt, minimum height=8pt] {#1}; }}

\begin{table}[t]
	\centering
	\setlength{\tabcolsep}{3pt}
	\begin{tabular}{lcccccc}
			\toprule
			{} &  optimist &  pessimist &  50\% &  singleton &  rival &  cost  \\
			\midrule
			optimist     &                   $-$ &                   \corr{0.87} &          \corr{0.98} &             \corr{0.98} &                       \corr{0.88} &         \corr{0.76}  \\
			pessimist     &                   \corr{0.87} &                   $-$ &          \corr{0.94} &             \corr{0.84} &                       \corr{0.50} &         \corr{0.73}  \\
			50\%             &                   \corr{0.98} &                   \corr{0.94} &          $-$ &             \corr{0.95} &                       \corr{0.78} &         \corr{0.79}  \\
			singleton          &                   \corr{0.98} &                   \corr{0.84} &          \corr{0.95} &             $-$&                       \corr{0.93} &         \corr{0.74}  \\
			rival &                   \corr{0.88} &                   \corr{0.50} &          \corr{0.78} &             \corr{0.93} &                       $-$&         \corr{0.63}  \\
			cost               &                   \corr{0.76} &                   \corr{0.73} &          \corr{0.79} &             \corr{0.74} &                       \corr{0.63} &         $-$  \\
			\bottomrule
	\end{tabular}
	\caption{Pearson Correlation Coefficients between measures for  $\eqphragmen$ (values near $1$ mean
		strong correlation).} 
	\label{tab:PCC}
\end{table}

\paragraph{Data.}
We conduct our experiments on all $551$ PB instances with approval votes from Pabulib \cite{fal-fli-pet-pie-sko-sto-szu-tal:c:pabulib} for which both $\phragmen$ and $\eqphragmen$  terminate within one second (on 1 thread of an Intel(R) Xeon(R) Gold 6338 CPU @ 2.00GHz core). 
In total, there are $3639$ losing projects for $\eqphragmen$ and $3513$ for $\phragmen$.

\paragraph{Measures.}
To compute the measures, we use the algorithms described in
\Cref{sec:measures}. For the brute-force and sampling algorithms, we
increase the approval score of the designated project by $1\%$ in each
step (repeating each step $100$ times for the sampling algorithms).
To simplify comparisons between measures, we normalize them to lie
between $0$ (being far away from victory) and $1$ (being close to
victory).  Specifically, for our four measures modifying a project's
approval score, we divide its original approval score by its approval
score plus the
measure. 
For example, if a project with score $20$ requires $80$ additional approvals (according to the measure), then the normalized value is $0.2$, since the project received $20\%$ of the needed approvals.
For $\costred$, we
divide $\costred(p)$ by $\cost(p)$, and for $\rivalred$, we report the
fraction of supporters who can continue to approve other projects.

\paragraph{Running Times.} We ran our experiments on 10 threads of an Intel(R) Xeon(R) Gold 6338 CPU @ 2.00GHz core. 
The algorithms for $\mathrm{optimist\hbox{-}add}$, $\mathrm{pessimist\hbox{-}add}$, $\mathrm{singleton\hbox{-}add}$, and $\mathrm{cost\hbox{-}red}$ are all very fast, finishing in below $20$ seconds on $95\%$ of instances. The sampling-based algorithms for $\mathrm{rival\hbox{-}red}$ and $\mathrm{50\%\hbox{-}add}$ are naturally slower, but still finish in $88\%$ of cases in below $10$ minutes. In sum, while our implementations are certainly not fully optimized, running the different algorithms for a PB exercise in one's city is feasible, even for tens of thousands of voters and hundreds of projects.

\subsection{Behavior and Correlation}
\label{be-co}
In \Cref{tab:PCC}, we show the linear correlation between our measures
for $\eqphragmen$, as given by the Pearson Correlation Coefficient
(PCC)
and in \Cref{app:pract-EQ}, we
show correlation plots.  In general, we want to stress that we observe
many projects close to getting funded under the different measures
(such projects can also be found in some of the plots presented in
this section).

\begin{figure}[t!]
    \centering 
\begin{subfigure}{0.3\textwidth}
\includegraphics[width=\textwidth]{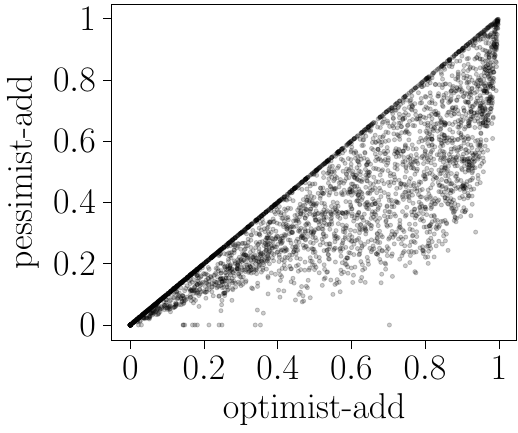}
\end{subfigure} \qquad
\begin{subfigure}{0.3\textwidth}
\includegraphics[width=\textwidth]{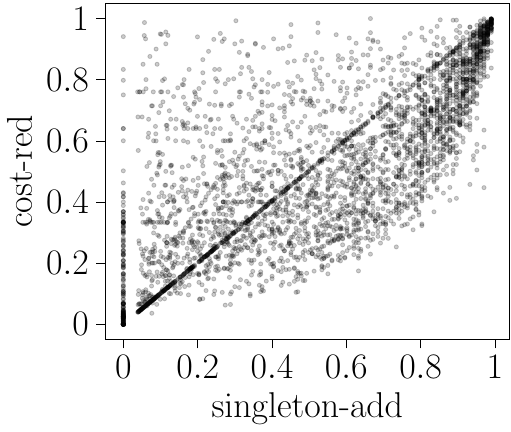}
\end{subfigure}
\caption{Correlation plots where each point is one project. Measures are normalized so that 1 denotes no change and 0 denotes a maximal-size change.} 
\label{fig:corr}
\end{figure} 

\paragraph{Add-Approvals and Singletons Measure.}
We first analyze the four measures related to adding approvals  to existing or new ballots. 
As seen in \Cref{tab:PCC}, they all have a strong pairwise correlation. 
However, there are small differences motivating a partitioning of the measures into two groups: 
$\optimistadd$, $\randomadd$, and $\singletonadd$ all have a pairwise correlation of at least $0.95$,  whereas $\pessimistadd$  has a lower correlation with the other three measures.
In the first group, the correlation between $\optimistadd$ and $\singletonadd$ is particularly strong.
For more than $90\%$ of the projects, the difference between the two measures is less than $0.051$ (exceptions include projects where one of the two measures is zero). 

Comparing the optimistic and pessimistic views on adding approvals, while they have a strong correlation of $0.87$, on the level of single projects, they can produce quite different results (see \Cref{fig:corr} left). In fact, for around $10\%$ of projects, almost twice as many approvals are needed to get the project funded under the pessimistic approach than under the optimistic one. Thus, under $\eqphragmen$, it really matters which voters add approvals for a project.

The $\randomadd$ measure lies between $\optimistadd$ and $\pessimistadd$ and is strongly correlated with them. It tends to be slightly closer to the optimistic view (average difference $0.063$) than to the pessimistic one ($0.077$).

\begin{figure}[t!]
    \centering 
\begin{subfigure}{0.33\textwidth}
\includegraphics[width=\textwidth]{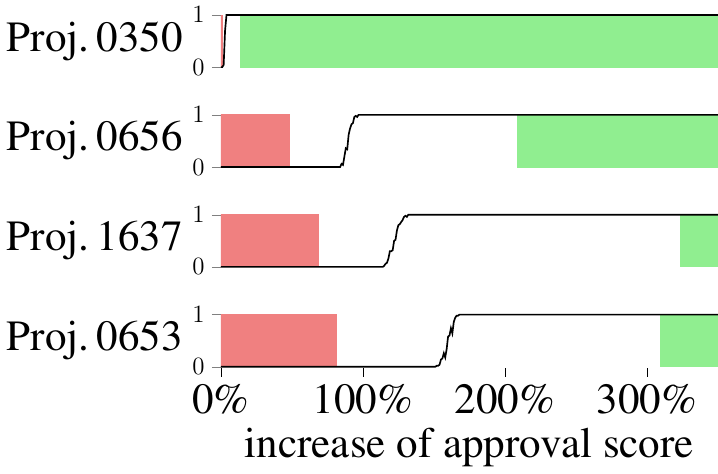}
  \caption{Wierzbno Wygledow 2019}\label{fig:50:a}
\end{subfigure}\qquad
\begin{subfigure}{0.33\textwidth}
\includegraphics[width=\textwidth]{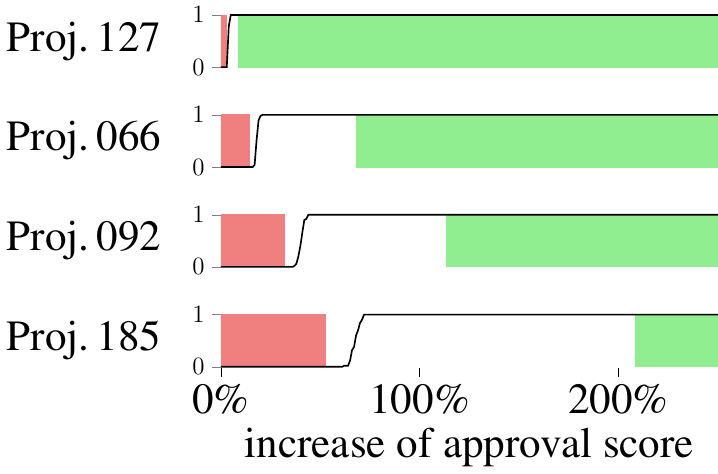}
  \caption{Lodz 2022}\label{fig:50:b}
\end{subfigure}
\caption{Line plots showing how the funding probability of a project develops from $0$ to $1$ when increasing its approval score by adding approvals uniformly at random to existing voters. The red area goes until the $\optimistadd$ value and the green area extends from the $\pessimistadd$ value.} 
\label{fig:50}
\end{figure} 

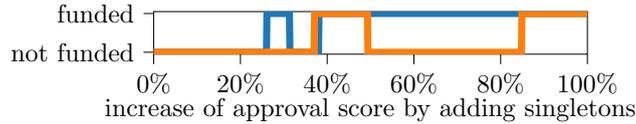
\begin{figure}[t]
    \centering
    \resizebox{0.55\textwidth}{!}{
    	\input{plots_MES/add_single.tex}
    }
    \caption{Behavior of two projects when adding voters who only support the project, taken from Warsaw 2023 (Praga-Polnoc, in blue) and Warsaw 2017 (Goclaw, in orange).}
    \label{fig:addsingle}
\end{figure}

In \Cref{fig:50}, 
we show  for two instances how the funding probabilities of projects evolve when adding approvals uniformly at random to existing voters. 
Projects very quickly transition from having a funding probability close to $0\%$ to having one close to $100\%$, even for projects where there is a large gap between $\optimistadd$ and $\pessimistadd$.
This phase-transition-like behavior appears for almost all projects. Thus, for practical purposes, reasonably optimistic and pessimistic views coincide (i.e., requiring that a project gets funded with at least $2\%$ probability is very similar to requiring that it gets funded with at least $98\%$ probability). Thus,  $\randomadd$ gives us a good estimate of how many approvals a project ``practically misses" to get funded.

If one wants to reduce the number of measures, it is probably simplest to use $\singletonadd$, as it is strongly correlated with $\optimistadd$ and $\randomadd$ and can be explained very easily.
Another advantage is that even for more complicated or slower rules, where efficient algorithms for the other measures might be hard to find, $\singletonadd$ can be computed via brute-force  (this approach is unsatisfying algorithmically, but if one is simply to provide information to participants of a PB election  once a year, then it typically suffices). 
The only drawback of this measure for $\equalshares$ is that, as discussed in \Cref{add:single}, adding singletons can make a project lose. 
In fact, this occasionally happens on our data in unexpected ways (see \Cref{fig:addsingle} for examples). 

\begin{figure}[t!]
    \centering 
\begin{subfigure}{0.4\textwidth}
\includegraphics[width=\textwidth]{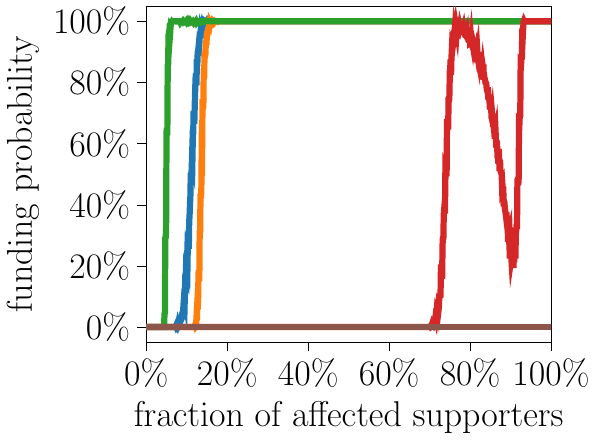}
  \caption{Boernerowo 2018}
\end{subfigure}
\qquad
\begin{subfigure}{0.4\textwidth}
\includegraphics[width=\textwidth]{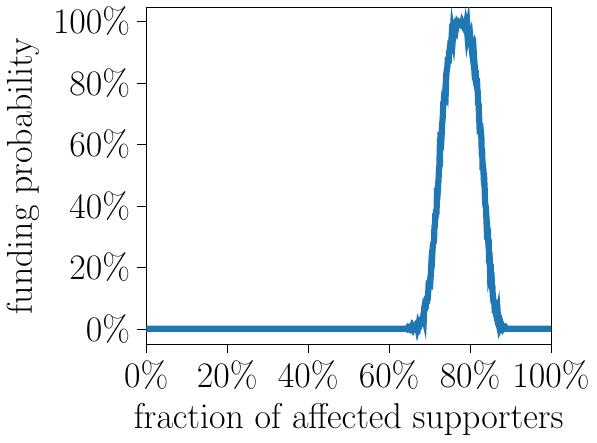}
  \caption{Kamionek 2019}
\end{subfigure}
\caption{Line plots showing how the funding probability of a project develops if we remove rivalry approvals from 
  its supporters selected uniformly at random (each line corresponds to a single project; all non-funded projects are shown). } 
\label{fig:rvex}
\end{figure} 

\begin{figure*}[t!]
    \centering 
\begin{subfigure}[t]{0.42\textwidth}
	\centering
\resizebox{0.9\textwidth}{!}{\input{plots_MES/scatterplot_MESCwiel_adding_single_changing_costs.tex}}
  \caption{Correlation plots for $\costred$ and $\singletonadd$. Larger font and pale background indicate the three projects that we focus on.}\label{fig:wiel:a}
\end{subfigure}\qquad
\begin{subfigure}[t]{0.42\textwidth}
	\centering
\includegraphics[height=4.6cm]{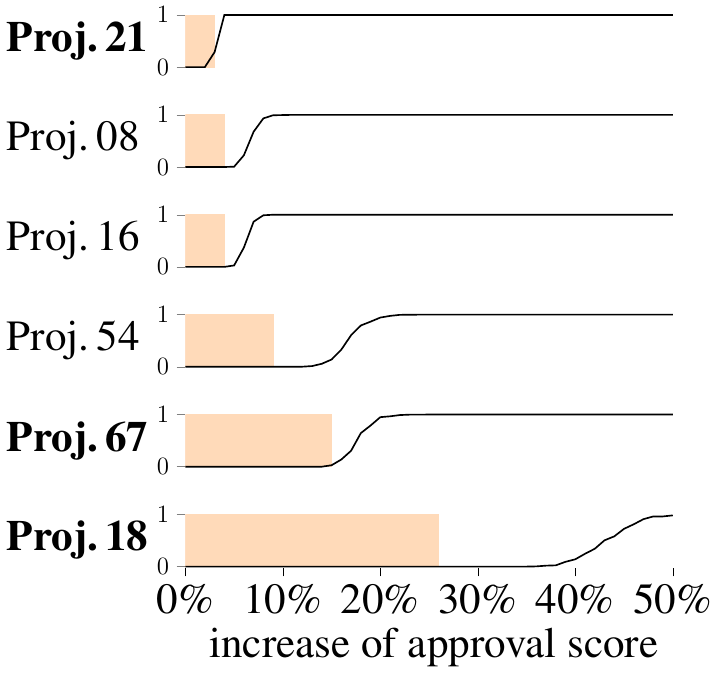}
  \caption{Funding probability when adding approvals uniformly at random to voters. Orange areas go to the $\singletonadd$ value.}\label{fig:wiel:b}
  \vspace{15pt}
\end{subfigure}

\begin{subfigure}[t]{0.42\textwidth}
\centering
\includegraphics[height=4.6cm]{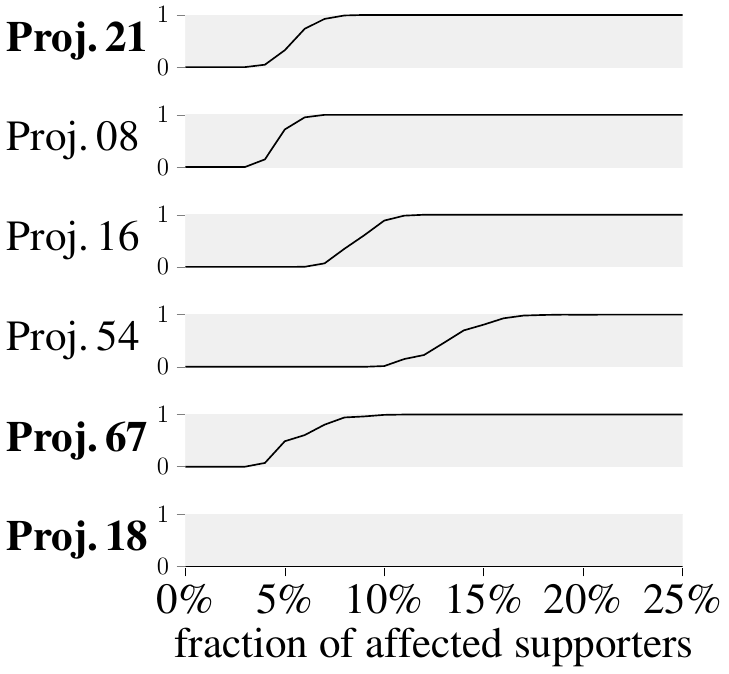}
  \caption{Funding probability when removing rivalry approvals from some uniformly at random selected group of supporters.}\label{fig:wiel:c}
\end{subfigure}
\qquad
\begin{subfigure}[t]{0.42\textwidth}
\centering
\vspace{-4.6cm}
\includegraphics[height=2.3cm]{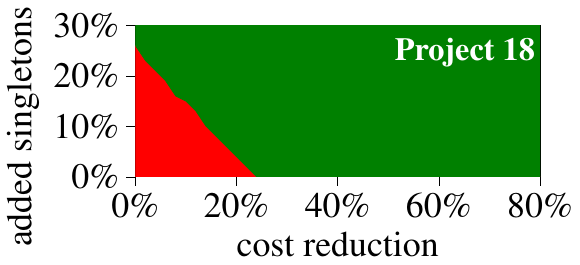}
\includegraphics[height=2.3cm]{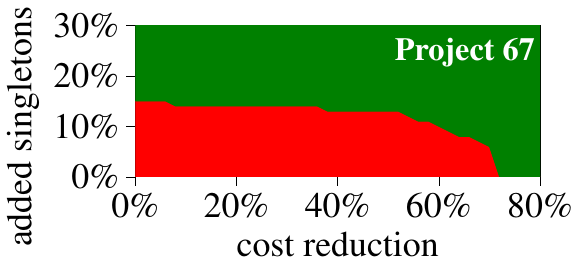}
  \caption{Tradeoff between fractional reduction of costs and increase of approval score by adding singletons.}\label{fig:wiel:d}
\end{subfigure}
\caption{Information on the  Wieliczka 2023 Green Million election.} 
\label{fig:wiel}
\end{figure*} 

\paragraph{Rivalry-Reduction Measure.}
Our sampling-based algorithm for $\rivalred$ returns a value for
$1605$ out of the $3639$ not-funded projects, which in particular implies that we can get these projects funded after removing all other approvals from the ballots of some (or all) of their supporters. 
This highlights the general power of lobbying one's supporters to only approve a single project in an election using a proportional rule.  
Regarding the correlation of $\rivalred$ with the other measures, 
\Cref{tab:PCC} shows the PCC correlation on the $1605$ projects where $\rivalred$ returns a value.
The correlation with the adding approval measures is strong (but not very strong). 
One interesting observation is that for projects where $\optimistadd$ is below $0.5$ (i.e., for projects whose approval score needs to at least double to be able to win), removing rivalry approvals is almost never sufficient to get funded.  
Regarding the $2034$ projects that cannot be funded after removing rivalry approvals, they have a very diverse performance with regard to the other measures.
For example, removing rivalry approvals might not be sufficient even for projects that are only missing a few approvals. 

Interestingly, removing more rivalry approvals does not necessarily
help a project.  The reason for this is that by removing the approval
of a rival, we can modify the execution of $\eqphragmen$ in arbitrary
rounds and thereby also help other projects to get funded.
Interestingly, this non-monotonic behavior of projects appears in many
forms.  \Cref{fig:rvex} shows two
examples. 
Nevertheless,
for rivalry reduction, for most projects there is also a quick jump from an almost $0\%$ to an almost $99\%$ funding probability. 

\paragraph{Cost-Reduction Measure.}
The cost reduction measure is less correlated to the other measures, which is expected since it is the only measure that modifies project costs. 
Yet, as all our measures help projects in some way, it is intuitive that there is a certain correlation (of around $0.75$) to the other measures. 
To analyze the tradeoff between adding approvals and reducing costs in more detail, \Cref{fig:corr} (right) shows the connection between $\costred$ and $\singletonadd$. 
We see that for a majority of projects, adding approvals is more powerful than reducing the project's cost. 
However, there are also numerous projects where it is the other way around, which makes it hard to make a general recommendation. 

\subsection{Wieliczka's Green Million} \label{wiel} We conclude by
constructing an information package for the PB election held in
Wieliczka in 2023, where
$\equalshares$ was used for the first time (\url{https://equalshares.net/resources/zielony-milion/}). This PB election focused
on ecological issues, with 64 projects placed on the ballot. 6586
people cast their votes. Each project cost up to 100\,000 PLN, and each
voter could approve any number of projects. The budget was one million
PLN (approximately 225\,000 EUR).
The vote data is available on \href{http://pabulib.org/}{Pabulib}.
Notably, Green Million used a more involved completion method of
$\equalshares$ which increases the initial balances of the voter's
bank accounts in small steps (see \Cref{app:prelim} for a description).  This is
computationally intensive, so we focus on the easier-to-compute
measures: $\costred$, $\singletonadd$, $\rivalred$, and $\randomadd$.

\Cref{fig:wiel} shows the results. In
particular, \Cref{fig:wiel:a} depicts the correlation between
$\singletonadd$ and $\costred$, \Cref{fig:wiel:b} how the funding
probabilities of projects change when adding approvals to existing
voters at random, and \Cref{fig:wiel:c} their behavior when removing
rivalry approvals.  Below, for three projects cherry-picked for interesting conclusions, we discuss the contents of a hypothetical information
package produced for the Green Million PB, and what project proposers could learn from it.
The analysis reinforces the conclusion that projects close to
being funded regularly appear and that our measures contribute
different perspectives.

\paragraph{Project 21 (cost 100\,000 PLN, 496 votes).} According to
our measures, this project was very close to winning.  Indeed, it needs
only about 3\% additional singleton votes or additional approvals to
be funded. It was also close to winning in terms of rivalry reduction, 
in the sense that relatively few of its supporters (below 10\%) would need to
refrain from supporting other projects in order for Project~21 to be
funded. On the other hand, the cost-reduction measure shows that the
project would have to be 20\% cheaper to be selected, given the
current votes. This indicates a strong project that
probably should be resubmitted in the next edition of the program,
with a slightly more aggressive support campaign. Reducing the cost of the project would be somewhat less effective.

\paragraph{Project 18 (cost 51\,000 PLN, 163 votes).}
The project performs similarly under $\singletonadd$ and
$\costred$; both measures indicate that the project was about $80\%$ on the way to winning. 
With respect to both of these measures, only 5 other losing projects do better.
However, Project~18's $\randomadd$ performance is
much worse, and the project continues to lose even if we remove
rivalry approvals from all its supporters.  This indicates that the
main problem of Project~18 is insufficient votes: Even in the first
round of $\equalshares$, its supporters don't have enough money available to buy the
project. Adding approvals to existing voters helps less with this
funding gap, as these voters might spend their budget in earlier
rounds on more popular projects.

\paragraph{Project 67 (cost 16\,500 PLN, 140 votes).}
Project~67 needs around $15\%$ more approvals according to
$\singletonadd$ and $\randomadd$, but its cost would need to be
reduced by $70\%$ to get funded.  Thus
Project 67 is not so far off in terms of the number of supporters it
has, but under $\equalshares$, those supporters spend almost all 
of their money on more popular projects
before Project~67 is considered by the rule. This
interpretation is confirmed by its performance with
respect to $\rivalred$, where its performance matches projects like
Project~21, which is very close to being funded according to all measures.
Hence, Project~67's main issue is competition with projects supported by the same voters.

\paragraph{Combined Strategy.}
One might also ask whether a combined strategy of slightly lowering
the cost of a project and getting a few more singleton voters could be
effective for our projects. We illustrate this approach in
\Cref{fig:wiel:d}.  The $x$-axis shows the percentage by which the
cost is reduced, and the $y$-axis shows the increase of the approval score of
the project (by adding singletons). A point is colored green if the project would 
be funded if both actions were performed, and red if it would continue to lose. 
We see that for Project~21 we can exchange additional
votes for cost reduction in a linear way, but for Project~67 the
behavior is highly nonlinear.

\section{Future Work}\label{sec:conclusions}
We have begun the work of designing useful information packages for voters and project proposers.
In future work, it would be useful to collect feedback from participants about what
measures and visualizations are of most interest.
There is also room for developing additional measures or improving some of the presented ones.
For instance, as an additional measure, one could consider changes in the overall budget, to be more precise, the minimum budget that is needed for the designated project to get funded. 
Moreover, for presented measures that increased the score of project $p$, one could try to prioritize
adding approvals to voters who are most likely to vote for $p$ 
(e.g., because they are similar
to current supporters of $p$).

\section*{Acknowledgments}
The main work was done while Niclas Boehmer was affiliated with TU
Berlin, where he was supported by the DFG project ComSoc-MPMS (NI
369/22).  This work was co-funded by the European Union under the
project Robotics and advanced industrial production
(reg. no. CZ.02.01.01/00/22\_008/0004590). Šimon Schierreich also
acknowledges the support of the Grant Agency of the Czech Technical
University in Prague, grant No. SGS23/205/OHK3/3T/18.  This project
has received funding from the European Research Council (ERC) under
the European Union’s Horizon 2020 research and innovation programme
(grant agreement No 101002854).  Piotr Skowron was supported by the
European Research Council (ERC; project PRO-DEMOCRATIC, grant
agreement No 101076570).

\begin{center}
   \includegraphics[width=3cm]{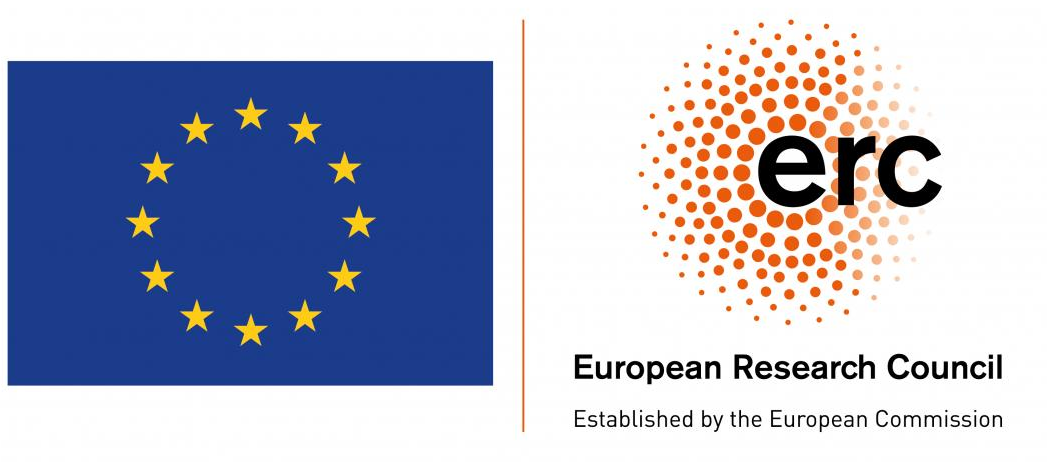}
\end{center}

\bibliographystyle{plainnat}

\clearpage

\appendix
\addtocontents{toc}{\protect\setcounter{tocdepth}{1}}

\section{Additional Material for \Cref{sec:prelim}}\label{app:prelim}
In this section, we describe the variant of $\equalshares$ used in Wieliczka's Green Million election in 2023. 

In the first round, the standard $\equalshares$ procedure is run. If the elected outcome is not exhaustive, the execution of $\equalshares$ is repeated with incremented values of voters' initial endowments (so that they are equal to $\nicefrac{B}{|V|}+1$, $\nicefrac{B}{|V|}+2$, $\ldots$ in further rounds). The procedure stops after the $i$th round if one of the following conditions is satisfied:
\begin{enumerate}
    \item An exhaustive (with respect to the budget $B$) outcome $W$ has been elected. Then $W$ is returned.
    \item The total cost of the elected outcome exceeds $B$. Then the outcome $W$ elected after the $(i-1)$th round is returned. Note that in this case, we know that $i > 1$ (since in the first round there are only $B$ money units in the system and the cost of every elected project is covered) and the cost of $W$ does not exceed $B$.
\end{enumerate}

At each step of the procedure, possible ties between projects are broken according to a randomly chosen priority ranking of projects.

\section{Additional Material for \Cref{sec:measures}}

\subsection{When Are Our Measures Undefined}\label{sec:undefined}

Below we list under what conditions are our measures defined:
\begin{description}
\item[Cost reduction.] For each of our rules, the cost-reduction measure is
  well-defined for projects that received at least a single approval
  (in the worst case we lower the cost to zero; for $\greedyav$ even
  the single approval is not necessary).

\item[Adding approvals.] The three adding-approvals measures
    (optimist, pessimist, and the randomized one) are defined in
    almost all cases: It suffices that all the voters approve a given
    project~$p$.  The exception happens if there is another project
    that already is approved by all the voters, the rule selects it
    prior to~$p$ (e.g., due to tie-breaking or the cost), and the
    remaining budget is insufficient for~$p$.

  \item[Adding singletons.] One can verify that the adding-singletons
    measure is well-defined for $\greedyav$ and $\phragmen$. For
    $\equalshares$, it is well-defined provided that project $p$ that
    we consider costs less than the available budget or it costs as much as the budget
    and is supported by all the voters.

\end{description}

\subsection{Cost-Reduction Measure}

\propCostRedPoly*
\begin{proof}
    Recall that the $\av$ rule considers projects in the order of their non-increasing approval scores (that is, the order depends solely on the number of approvals and is not affected by the cost of the projects) and a project is included in the outcome~$W$ if doing so will not exceed the available budget. Let~$i\in\mathbb{N}$ be the round in which the rule considers project~$p$ and~$B'$ be the remaining budget at the beginning of round~$i$. Since~$p$ was not initially funded, we have~$\cost(p) > B'$, however, setting~$\cost(p) = B'$ makes the project affordable. Therefore, $\costred^\av_E(p) = B'$.

    For $\ph$, we first ignore the project~$p$ and execute the rule with the remaining candidates. This rule in~$r$ time steps adds some projects to the outcome. For each time step~$t_1,\ldots,t_r$, we compute the endowments $E_1,\ldots,E_r$ of supporters of our distinguished candidate~$p$. The result is then the maximum~$x_i$ on all time steps $i\in[r]$ in which $\cost(W) + x_i \leq B$, where $x_i = \min\{E_i, B - \cost(W)\}$. Suppose that $x$ is the maximum possible cost so that even after setting $\cost(p) = x$, the project $p$ is funded. According to the definition of the rule, the project is funded in some time step $t_i$, $i\in[r]$. If $x > B - \cost(W)$, then~$p$ cannot be added to the outcome since there is not enough remaining budget. If $x > E_i$, then~$p$ cannot be founded in this round, as the endowments of supporters of $p$ are not high enough. Therefore, $x$ is at most $\min_{j\in[r]}\{E_j,B-\cost(W)\}$, which is exactly the result of our algorithm.

    Finally, for $\eq$, we again examine the execution of the rule without the project~$p$. We compute the value for each iteration separately and then return the maximum. 
    Assume that in round $i$ project $d$ is funded as it is $q$-affordable. 
    Then we want to find a cost value $x$ such that ${\sum_{v_j\in A(c)}\min(b_i(v_j),q\cdot x)= x}$. 
    Observe that $\sum_{v_j\in A(c)}\min(b_i(v_j),q\cdot x)$ is monotonically increasing and piecewise linear. 
    Let $b^1\leq \dots \leq b^n$ be the budget values of the voters.
    Then for $x\in \left[\frac{b_i}{q},\frac{b_{i+1}}{q}\right]$ the function is linear as $(n-i)\cdot q \cdot x + \sum_{j=1}^i b_j$. 
    We can first find the value of $b_i$ such that the solution lies within $\left[\frac{b_i}{q},\frac{b_{i+1}}{q}\right]$  and then simply check within this interval. 
    If the solution lies in no such interval, then we know $x>b^n$  or $x<b^1$. We can handle these cases.
\end{proof}

\subsection{Add-Approvals Measure}

\subsubsection{Optimistic Add}
\proOptimistAddPoly*
\begin{proof}
    We start with the $\av$ rule. Crucial observation in this case is that for the rule, it does not matter where we add approvals. So, we run the rule on the original instance as usual. The interesting part of the process takes place once the funding of the currently processed project $c$ first decreases the value of the remaining budget below $\cost(p)$. Before funding $c$, we compute $m^*$ as $|A(c)| - |A(p)| + [c \succ p]$, where $[ c \succ p ]$ evaluates to $1$ if $c$ precedes $p$ in tie-breaking order and to $0$ otherwise. If $m^* + |A(p)| > |V|$, then the measure is not defined, and we terminate. Otherwise, the algorithm outputs the value of $m^*$. For the correctness, suppose that $\optimistadd^\av_E(p) = m' < m^*$. If we add $m'$ additional voters that approve $p$, then $p$ is assumed by the rule in some round after $c$ is funded. However, $c$ was selected so that, after funding $c$, the remaining budget is below $\cost(p)$. Consequently, adding $m'$ additional voters is not sufficient to fund $p$ and therefore $\optimistadd_E^\av(p) \geq m^*$. Since, by the definition of the rule, the distinguished project $p$ is funded after adding $m^*$ approvals, the algorithm is correct and clearly performs only $O(1)$ additional operations per round.

    For $\ph$, the high-level idea is to examine each round separately and return the minimum number of additional approvals needed to fund $p$. More formally, let $i$ be a round such that there is still enough money to fund $p$ (otherwise, we stop the execution of the rule and return the currently stored minimum). Let $f_i$ be the money that the distinguished candidate lacks in order to be funded in round $i$, and let $X\subseteq V$ be the set of voters not approving $p$. We sort the voters in $X$ according to their current balance and select the minimum possible number $m_i$ of voters (in descending order according to their balances) such that the current endowment of $p$ plus the summed endowments of selected voters from $X$ is at least $\cost(p)$. The outcome of the algorithm is then the minimum for all $m_i$. In each round, we additionally sort the voters in~$X$ according to their current balance, which can be done in $O(n\log n)$ time, where $n=|V|$.
     
    Finally, with $\eq$ we look again at each round where there is still enough money left to buy $p$, compute the number of additional approvals we need to fund $p$ in this round, and return the minimum over these values for all the examined rounds. Let $i$ be a round currently processed, $q_i$ be the minimum $q$ such that a project is $q$-affordable in round $i$, and $f_i$ be the missing money to fund $p$. Additionally, we set $X\subseteq V$ as the set of voters that do not approve $p$ and $e^i_v$ as the endowment of the voter $v\in X$ in the round $i$. We sort the voters in $X$ in descending order according to their endowments; assume that $e^i_{x_1}\geq \dots \geq e^i_{x_n}$. Now, we find the minimum $m_i$ such that $\sum_{t}^{m_i} \min(e^i_{x_t},\cost(p)\cdot q_i) \geq \cost(p)$. If there are still enough funds in the final round, we do the same thing with $q_i=\infty$. Is in the case of $\ph$, in every round we additionally sort all voters in $X$ and do a linear-time computation to determine $m_i$. That is, for $\eq$, we again need $O(n\log n)$ additional operations per round.
\end{proof}

\subsubsection{Pessimist Add}

\thmphragmenpessimist*
\begin{proof}
  Membership in $\conp$ is clear: Given an instance of a problem it
  suffices to guess $\ell$ voters and add the desiginated project to
  their approval sets. Then, accept on computation paths where the
  project is funded and reject on those where it is not. If
  $\pessimistadd^\ph_E(p) \leq \ell$, then all paths accept and
  otherwise at least one path rejects.

  To show $\conp$-hardness, we will work with a complement of our
  problem: Given a PB instance $E$, a losing project $p$, and
  nonnegative integer $\ell$, we ask if it is possible to add $\ell$
  approvals for $p$ without getting it to be funded. We show that this
  problem is $\np$-hard by giving a reduction from $\rxthreec$, a
  restricted variant of the \textsc{Exact Cover by 3-Sets}
  problem~\cite{gon:j:x3c}. In this problem we are given a universe
  $U = \{u_1, u_2, \ldots, u_n\}$ of elements and a collection
  $\calS = \{S_1, S_2, \ldots, S_n\}$ of size-$3$ subsets of $U$.
  Each element $u_i \in U$ appears in exactly three sets from $\calS$
  and $n$ is divisible by three.  We ask whether there exists an exact
  cover over $U$, that is, a collection of sets from $S$ such that
  each element from $U$ appears in exactly one set.  In this case,
  since each set has size three, each valid exact cover must have size
  $\nicefrac{n}{3}$.

  W.l.o.g., we assume that $n \geq 60$ and $n = 54q + 6$ for some
  $q \in \naturals$.  Otherwise, we could keep adding three new
  elements $e_1, e_2, e_3$ and three new (identical) sets
  $\{e_1, e_2, e_3\}, \{e_1, e_2, e_3\}, \{e_1, e_2, e_3\}$ until our
  condition were met (we interpret $\calS$ is a multiset).  These
  newly added sets do not interfere with the original ones, so the
  solution for the modified instance always consists of a solution for
  the original one plus one copy of each newly created set.

  \paragraph{Construction.}
  We form a PB instance with project set
  $P = \{p\} \cup U \cup D \cup \{b_1, b_2\}$, where $p$ is the
  designated project, $U$ is the universe (so each of its members
  doubles as a project), $D = \{d_1, \ldots, d_n\}$ is a set of dummy
  candidates, and projects $b_1$ and $b_2$ are used to initiate
  appropriate bank-account balances for the voters. Each project has
  the same unit cost.  To introduce the voters, we need the following
  notation:
  \begin{align*}
    w &= \nicefrac{1}{18}(11n^{10}+3n^9), \\
    z &= \nicefrac{1}{18}(5n^{10}+3n^9), \\
    y &= \nicefrac{1}{18}(4n^{10}-3n^9), \\
    \alpha &= \nicefrac{1}{27}(4n^{10}-3n^9-6n-18), \text{ and}\\
    s &= \nicefrac{n^9}{6}.
  \end{align*}
  Numbers $z, y, w$ and $s$ are positive integers because $n$ is
  divisible by $6$.  Further:
  \begin{align*}
    \alpha   & = \nicefrac{1}{27}(4n^{10}-3n^9-6 \cdot (54q+6)-18)\\
             &= \nicefrac{1}{27}(4n^{10}-3n^9-12 \cdot 27q-2 \cdot 27),
  \end{align*}
  so $\alpha$ also is a positive integer. Finally, we see that
  $\alpha < y$.  We create four groups of voters as follows:
  
  \begin{description}
  \item[Group $\boldsymbol W$:] We have $w$ voters approving only
    $b_1$ and $b_2$.

  \item[Group $\boldsymbol Z$:] We have $z$ voters approving candidates
    $\{p\} \cup U$.
    
  \item[Group $\boldsymbol Y$:] We have $y$ voters approving candidates
    $D \cup \{b_1, b_2\}$. Further, first $\alpha$ of them also
    approve $p$ (we refer to them as group $A$; naturally, the
    $A$-voters also belong to group $Y$).
  \item[Group $\boldsymbol S$ (Set Voters):] For each $S_i \in \calS$,
    we create $s$ voters approving candidates
    $D \cup \{b_1, b_2\} \cup (U \setminus S_i)$. That is, each of
    them approves $n + 2 + (n - 3) = 2n-1$ candidates (recall that
    $|S_i| = 3)$. There are $s \cdot n$ set voters in total and each
    universe candidate is approved by exactly $s \cdot (n-3)$ of them.
  \end{description}
  Finally, we set the budget to be $B = 2n+2$, and we set the
  tie-breaking order to be:
  \[
     b_1 \succ b_2 \succ p \succ u_1 \succ d_1 \succ u_2 \succ d_2 \succ
     \cdots \succ u_n \succ d_n.
  \]
  We ask if we can add $\ell = \nicefrac{n}{3}$ approvals in such a
  way that~$p$ does not get funded.  We introduce the
  following constants:
  \begin{align*}
    t_1 = \nicefrac{1}{n^{10}} &,& t_2 = \nicefrac{2}{n^{10}},
  \end{align*}                                 
  representing time periods after which we buy some candidates (this
  will become clear a bit later).

  \paragraph{Intuitive Idea.}
  Let us now explain the intuition behind our construction. For the
  sake of simplicity, let us assume that approvals can only be added
  to the set voters. Later we will show why this assumption is valid.

  Consider our PB instance after $\ell$ extra approvals for $p$ were
  added to some of the set voters. $\phragmen$ proceeds as follows:
  First, at time $\nicefrac{t_2}{2}$ the rule selects $b_1$ and at
  time $t_2$ the rule selects $b_2$. Consequently, at this moment each
  of the $Z$-voters has budget $t_2$ and all other voters have budget
  $0$.  Then, there are $n$ \emph{iterations} which consist of
  selecting members of $U$ and $D$ in an interleaved fashion.  At the
  beginning of the $i$-th iteration, each $Z$-voter has budget $t_2$,
  whereas the set voters and the $Y$-voters have budgets $0$. The
  $W$-voters do not play any role since all their projects have
  already been purchased.  After time $t_1$ (within the iteration)
  project $u_i$ is purchased by the voters from group $Z$ and the set
  voters corresponding to sets that do not include $u_i$.  Project $p$
  cannot be purchased at this moment regardless of how we add the
  approvals.  Then, after time $t_2$ since purchasing $u_i$, the
  voters either buy $d_i$ or $p$ (the latter happens if we added
  approvals to at least two set voters that did not pay for $u_i$,
  i.e., to two set voters who correspond to sets that include $u_i$).
  If the voters purchase $d_i$ instead of $p$, then all the $Y$-voters
  and all the set voters have their budgets reset to $0$, whereas each
  $Z$-voter has budget $t_2$. Hence the situation is the same as at
  the beginning of the iteration. Then the $(i+1)$-th iteration starts.

  After $n$ iterations we run out of budget. If $p$ were not selected,
  then it means that for each iteration, after buying a universe
  project we purchased a project from $D$. This means that we must
  have added approvals for $p$ to set voters that correspond to an
  exact cover of $U$ (the reverse direction is also immediate: If we
  add approvals to set voters corresponding to an exact cover then $p$
  is not purchased).  In other words, $p$ can lose after adding
  $\ell=\nicefrac{n}{3}$ approvals if and only if there exists an
  exact cover in our $\rxthreec$ instance.

  \paragraph{Formal Argument (Approvals Added to Set Voters Only).}
  Now we formally argue that the execution of the rule proceeds as
  described above, assuming that all the additional approvals for $p$
  went to the set voters (below we give a number of (in)equalities; we
  give the somewhat tedious explanation as to why the hold at the end
  of the proof):
  \begin{description}
  \item[Initialization.] Altogether, there are $y + s \cdot n + w$
    voters who approve $b_1$ (and $b_2$). Thus, after time
    $\nicefrac{t_2}{2}$ since the beginning of the execution of
    $\phragmen$, voters supporting~$b_1$ have enough funds to purchase
    it (and, analogously, after the following $\nicefrac{t_2}{2}$
    chunk of time they have enough funds to also purchase
    $b_2$). Indeed, we have that:
    \[
      (y + s \cdot n + w) \cdot \nicefrac{t_2}{2} = 1.
    \]
    All the other projects have far fewer approvals, so~$b_1$
    and~$b_2$ are purchased first. Consequently, at time $t_2$ since
    the beginning of the execution, $b_1$ and $b_2$ are funded.

  \item[Budgets at the beginning of each iteration.]  Assuming that
    project $p$ has not been funded, at the beginning of each
    iteration each $Z$-voter has budget $t_2$ and all the other voters
    have budgets $0$.  This is exactly the situation right after
    project $b_2$ is selected, when the first iteration starts.

  \item[Iteration~$\boldsymbol i$ until time $\boldsymbol{t_1}$.] Let
    us consider what happens during the $i$-th iteration until time
    $t_1$ since its beginning. First, as witnessed by the following
    equality, at time~$t_1$ voters supporting $u_i$ have enough funds
    to purchase it:
    \[
      z \cdot (t_1 + t_2) + s \cdot (n-3) \cdot t_1 = 1.
    \]
    (Indeed, each $Z$-voter had $t_2$ money at the beginning of the
    round and obtained another $t_1$ until time $t_1$; further, each
    of the $s \cdot (n-3)$ set voters that approve $u_i$ had budget
    $0$ at the beginning of the round and earned $t_1$ until this
    moment.)  Naturally, voters supporting projects
    $u_{i+1}, \ldots, u_n$ also have sufficient amounts of money to
    buy them, but they lose to $u_i$ due to the tie-breaking order
    (projects $u_1, \ldots, u_{i-1}$ were purchased in prior
    iterations\footnote{Indeed, formally we should express our
      reasoning as an inductive proof, but we believe that the current
      approach is sufficiently clear and a bit lighter.}).  It remains
    to argue that neither $p$ nor any of the $D$-projects can be
    purchased at time $t_1$ or prior to it. For the case of~$p$, we
    see that at time $t_1$ its voters have total budget:
    \[
      z \cdot (t_1 + t_2) + \ell \cdot t_1 + \alpha \cdot t_1 < 1
    \]
    which does not suffice to buy $p$.  Finally, at time $t_1$ since
    the beginning of the iteration, each $D$-project is supported by
    voters who in total have budget:
    \[
      y \cdot t_1 + s \cdot n \cdot t_1 < 1
    \]
    and, so, neither of the $D$-projects can be purchased.
    Consequently, at time $t_1$ since the begining of the iteration,
    project $u_i$ is funded and its supportes have their bank accounts
    reset to $0$.

  \item[Iteration $\boldsymbol i$ from time $\boldsymbol t_1$ until
    time $\boldsymbol{t_1+t_2}$.]  After time $t_1+t_2$ since the
    beginning of the iteration, voters approving project $d_i$ have
    enough funds to purchase it. Indeed, we have:
    \[
      y \cdot (t_1 + t_2) + s \cdot n \cdot t_2 + 3s \cdot t_1
      = 1
    \]
    (note that by this time the $Y$-voters earned $t_1+t_2$ each, each
    set voter earned $t_2$ since time $t_1$, and the budgets of the
    $3s$ set voters corresponding to sets that include $u_i$ were not
    reset to zero at time $t_1$, so they still have the money they
    earned since the beginning of the iteration).  Analogously as in
    the preceding paragraph, projects $d_{i+1}, \ldots, d_n$ also are
    supported by voters who can afford them at time $t_1+t_2$ but
    $d_i$ wins due to the tie-breaking order (and, assuming that $p$
    were not selected yet, the other $D$-projects were purchased in
    previous rounds).\smallskip

    At time $t_1+t_2$ since the beginning of the iteration, each of
    the not-yet-selected universe projects is supported by voters who
    in total have at most:
    \[
      z \cdot t_2 + s \cdot (n-3) \cdot t_2 + 3s \cdot t_1 < 1.
    \]
    money, so neither of them can be purchased at this time.\smallskip

    Finally, let us consider the amount of money that voters
    supporting $p$ have at time $t_1+t_2$. To this end, let $g_i$ be
    the number of set voters whose corresponding sets contain~$u_i$
    (and, hence, who did not spend their money since the beginning of
    the iteration) and who also got additional approvals for
    $p$. Consequently, at time $t_1+t_2$ since the beginning of the
    iteration, voters approving $p$ have the following amount of money
    (assuming that $p$ was not purchased prior to this time):
    \[
      f(g_i) = z \cdot t_2 + \ell \cdot t_2 + \alpha \cdot (t_1 + t_2) + g_i
      \cdot t_1.
    \]
    Indeed, each $Z$-voter earned $t_2$ since time $t_1$, each set
    voter who got an approval for $p$ also earned $t_2$ since
    time~$t_1$, each $A$-voter earned $t_1+t_2$ since the beginning of
    the iteration, and each of the $g_i$ set voters who got an
    approval for $p$ and who corresponds to a set including~$u_i$ did
    not spend its $t_1$ amount of money at time $t_1$. We observe
    that:
    \[
       f(1) < 1 \leq f(2).
    \]
    Hence, if $g_i = 1$ then $\phragmen$ selects $d_i$ at time
    $t_1+t_2$ since the beginning of the iteration, and we start
    iteration $i+1$ (with each $Z$-voter having $t_2$ money and each
    $Y$ and each set voter having zero money).  Yet, if $g_i \geq 2$,
    then $\phragmen$ selects $p$ (possibly even earlier than time
    $t_1 + t_2$, especially if $g_i > 2$).    
  \end{description}
  All in all, it is possible to add $\ell$ approvals for $p$ to the
  set voters without getting $p$ to be funded if and only if our input
  $\rxthreec$ instance has an exact cover.

  \paragraph{Adding $\textbf{p}$-Approvals Beyond Set Voters.}
  Next, we explain why it suffices to focus on adding $p$-approvals to
  the set voters.  We make the following observations:
  \begin{enumerate}
  \item If we add a $p$-approval to even a single $W$-voter (and the
    remaining $\ell-1$ approvals to whatever other voters) , then at
    latest at time $t_1+t_2$ since the beginning of the second
    iteration, this voter would have accumulated sufficient amount of
    money that $p$ would be selected instead of $d_2$.
  \item All $Z$-voters already approve $p$, so it is impossible to add
    further $p$-approvals to them.
  \item If there is a solution where we add $p$-approvals to some
    $Y$-voters (who are not $A$-voters, as those already approve $p$),
    then $p$ also is not funded if we add the approvals to set voters
    instead. Indeed, the $Y$-voters have at least as much money as the
    set voters at each point of time during each iteration.
  \end{enumerate}
  This justifies why it suffices to consider adding approvals for $p$
  to the set voters only. However, regarding the $Y$-voters we can say
  something stronger: If we add an approval to some $Y$-voter $y_j$,
  then in each iteration $i$, when considering $d_i$ vs $p$, voter
  $y_j$ would contribute $t_1+t_2$ to buying $p$ (since $y_j$ does not
  pay for $u_i$), whereas any $S$-voter $s_k$ contributes $t_1+t_2$
  only if $u_i$ belongs to the corresponding set (otherwise $s_k$
  contributes to $u_i$ and is left with $t_2 < t_1+t_2$ budget for
  $p$).  Thus, in iteration $i$ at time $t_1+t_2$ $p$-supporters would
  have budget at least:
  \[
    z \cdot t_2 + \ell \cdot t_2 + \alpha \cdot (t_1 + t_2) + g_i \cdot
    t_1 + y_p \cdot t_1,
  \]
  where $y_p \geq 1$ is the number of $Y$-voters that got approvals
  for $p$ ($g_i$ was defined when analyzing the second part of the
  iteration).  It is not hard to see that $y_p \geq 2$ results in
  selecting $p$ before $d_1$ and $y_p = 1$ results in selecting $p$
  before $d_h$, where $h$ is the index of element $u_h$ whose one of
  the voters corresponding to sets containing $u_h$ received an
  approval towards $p$.  Therefore, adding a $p$-approval to any of
  the $Y$-voters (who is not a $A$-voter) results in selecting $p$.

  Consequently, if there is a solution for our problem, then it
  consists of adding $p$-approvals to set voters only. This means that
  the number of solutions for our problem is equal to the number of
  exact covers in the input $\rxthreec$ instance.

  \paragraph{Calculations.}
  To complete the proof, we establish that the (in)equalities that we
  assered actually hold.
  \begin{enumerate}
  \item $1 = (y + s \cdot n + w) \cdot \nicefrac{t_2}{2}$.  This
    equation asserts that after time $\nicefrac{t_2}{2}$, candidate
    $b_1$ can be bought (and later analogously candidate $b_2$).  This
    is true because
    $(y + s \cdot n + w) \cdot \nicefrac{t_2}{2} =
    (\frac{4n^{10}-3n^9}{18} + \nicefrac{n^9}{6} \cdot n +
    \frac{11n^{10}+3n^9}{18}) \cdot \nicefrac{1}{n^{10}} =
    \frac{4n^{10}-3n^9+3n^{10}+11n^{10}+3n^9}{18n^{10}} =
    \frac{18n^{10}}{18n^{10}} = 1$.
  \item $1 = z \cdot (t_1 + t_2) + s \cdot (n-3) \cdot t_1$.  This
    equation asserts that $U$-candidate $u_i$ can be selected after
    time $t_1$ provided that $Z$-voters had initially budget $t_2$.
    One can verify that
    $\frac{1 - s \cdot (n-3) \cdot t_1}{t_1 + t_2} = \frac{n^{10} -
      \nicefrac{n^9}{6} \cdot (n-3) \cdot 1}{1 + 2} =
    \frac{5n^{10}+3n^9}{18} = z$.
  \item
    $1 = y \cdot (t_1 + t_2) + s \cdot n \cdot t_2 + 3 \cdot s \cdot
    t_1$.  This equation asserts that $D$-candidate $d_i$ can be
    selected after time $t_1+t+2$ from the beginning of an iteration
    provided that formerly $U$-candidate $u_i$ was selected at time
    $t_1$.  One can verify that
    $\frac{1 - s \cdot n \cdot t_2 + 3 \cdot s \cdot t_1}{t_1 + t_2} =
    \frac{n^{10} - \nicefrac{n^9}{6} \cdot n \cdot 2 -
      \nicefrac{n^9}{6} \cdot 3 \cdot 1}{1 + 2} =
    \frac{4n^{10}-3n^9}{18} = y$.  It also implies that
    $1 - z \cdot (t_1 + t_2) = s \cdot (n-3) \cdot t_1$
  \item
    $z \cdot t_2 + \ell \cdot t_2 + \alpha \cdot (t_1 + t_2) + t_1 < 1
    \leq (z \cdot t_2 + \ell \cdot t_2 + \alpha \cdot (t_1 + t_2) +
    t_1) + t_1$.  This inequality asserts for each $U$-candidate $u_i$
    that if we add approval to one voter corresponding to a set
    containing $u_i$, then $p$ still loses with $d_i$ at time $t_2$,
    but if we add approvals to at least two voters corresponding to
    sets containing $u_i$, then $p$ will be selected (due to
    tie-breaking at $t_1+t_2$ for two or sooner for at least three).
    Please note that it is equivalent to saying that
    $\alpha \in [\frac{1 - (z \cdot t_2 + \ell \cdot t_2 +
      2t_1)}{t_1+t_2}, \frac{1 - (z \cdot t_2 + \ell \cdot t_2 +
      t_1)}{t_1+t_2})$.  Since
    $\alpha = \frac{4n^{10}-3n^9-6n-18}{27} =
    \frac{9n^{10}-5n^{10}-3n^9-6n-18}{27} =
    \frac{n^{10}-\frac{5n^{10}+3n^9}{18} \cdot 2 - \nicefrac{n}{3} -
      2}{1 + 2} = \frac{1 - z \cdot t_2 - \nicefrac{n}{3} \cdot t_2 -
      2 \cdot t_1}{t_1 + t_2}$, $\alpha$ lies in this range.
  \item $y \cdot t_1 + s \cdot n \cdot t_1 < 1$.  This inequality
    guarantees at in iteration $i$ at time $t_1$, each yet-unbought
    $D$-candidate will have too low budget of its supporters to be
    bought so we will prefer to purchase $u_i$.  This also holds as
    $y \cdot y_1 + s \cdot n \cdot t_1 = \frac{\frac{4n^{10}-3n^9}{18}
      + \nicefrac{n^9}{6} \cdot n}{n^{10}} =
    \frac{4n^{10}-3n^9+3n^{10}}{18n^{10}} =
    \frac{7n^{10}-3n^9}{18n^{10}} < 1$.
  \item $z \cdot (t_1 + t_2) + \ell \cdot t_1 + \alpha \cdot t_1 < 1$.
    This inequality assures that in iteration $i$ at time $t_1$, $p$
    will have too low budget of its supporters to be bought so we will
    prefer to purchase $u_i$.  Let us observe that
    $z \cdot (t_1 + t_2) + \nicefrac{n}{3} \cdot t_1 + \alpha \cdot
    t_1 < 1 \iff \nicefrac{n}{3} \cdot t_1 + \alpha \cdot t_1 < 1 - z
    \cdot (t_1 + t_2) \iff \nicefrac{n}{3} \cdot t_1 + \alpha \cdot
    t_1 < s \cdot (n-3) \cdot t_1 \iff \nicefrac{n}{3} + \alpha < s
    \cdot (n-3) \iff s \cdot (n-3) - \nicefrac{n}{3} - \alpha > 0$.
    Since
    $s \cdot (n-3) - \nicefrac{n}{3} - \alpha = \nicefrac{n^9}{6}
    \cdot (n-3) - \nicefrac{n}{3} - \frac{4n^{10}-3n^9-6n-18}{27} =
    \frac{9n^{10}-27n^9-18n-8n^{10}+6n^9+12n+36}{54} =
    \frac{n^{10}-21n^9-6n+36}{54} = \frac{(n-22) \cdot n^9 + (n^8-6)
      \cdot n +36}{54} > 0$ for $n \leq 60$, this inequality also
    holds.
  \item
    $z \cdot t_2 + s \cdot (n-3) \cdot t_2 + 3 \cdot s \cdot t_1 < 1$.
    This inequality asserts that in iteration $i$ at time $t_1+t_2$,
    each yet-unbought $U$-candidate will have too low budget of its
    supporters to be bought so we will purchase $d_i$ or $p$.  If we
    place the actual values in the left side, we obtain that
    $z \cdot t_2 + s \cdot (n-3) \cdot t_2 + 3 \cdot s \cdot t_1 =
    \frac{\frac{5n^{10}+3n^9}{18} \cdot 2 + \nicefrac{n^9}{6} \cdot
      (n-3) \cdot 2 + 3 \cdot \nicefrac{n^9}{6} \cdot 1}{n^{10}} =
    \frac{10n^{10}+6n^9 + 6n^{10}-18n^9 + 9n^9}{18n^{10}} =
    \frac{16n^{10}-3n^9}{18n^10} < 1$.
  \end{enumerate}
  This completes the proof.
\end{proof}

\thmFPT*
\begin{proof}
    We start with \phragmen{} and later show how the algorithm can be tweaked to work also for \equalshares{}.  Let $E = (P,V,B)$ be a PB  instance with losing project $p$. We will show how to compute the largest number of approvals whose addition \emph{does not} lead to funding $p$. The value of $\pessimistadd_E(p)$ is one larger. Let $k = |\phragmen(E)|$ be the number of funded projects or, equivalently, the number of rounds performed by the rule. We assume that $p$ is approved by at least one voter in $V$.

    For each round $i$, let $m_i$ be the difference between $\cost(p)$ and the amount of funds that voters approving $p$ have in round~$i$. Our goal is to select the largest group of voters who do not approve $p$ and whose total funds in each round $i$ are at most~$m_i$. 
  
    For each voter $v_i \in V \setminus A(p)$, we define its \emph{balance vector} $(b^i_1,\ldots, b^i_k)$, which contains the balance of the voter's bank account right before each round. We partition the voters not approving $p$ into \emph{voter-types} $T=\{T_1,\ldots,T_t\}$, where each type consists of voters with identical balance vectors.

    To solve the problem, we form an integer linear program (ILP). For every voter-type $T_i$ we have a nonnegative integer variable $x_{T_i}$ that represents the number of voters of this type that will additionally approve~$p$. For each round $i$ we form the following \emph{round-constraint} (using strict inequality, if needed due to tie-breaking):
    \[
        \textstyle \sum_{j\in\{1, \ldots, t\}} x_{T_j}\cdot b^{T_j}_i \leq m_i.
    \]
    The objective function is to maximize the sum of all the $x_{T_i}$ variables. Our algorithm outputs the value of the objective function plus one.

    Finally, we note that there are at most $O(2^k)$ voter types (indeed, each voter type corresponds to a $k$-dimensional $0$/$1$ vector, which has $1$ in position $i$ if a voter approves---and, hence, pays for---the candidate selected in round $i$ or not). Thus, the number of variables in our ILP is $O(2^k)$ and we can solve it using the classic algorithm of \citet{len:j:integer-fixed} in $\fpt$ time with respect to $k$.

    For $\equalshares$, the algorithm is almost the same. We just extend the preprocessing by computing values $q_1,\ldots,q_k$ that stand for the value of $q(c_i)$ used to fund project~$c_i$ in round $i\in[k]$ and we additionally extend the definition of balance vector by the value $b_{k+1}^i$. In the ILP, we then replace each round-constraint for each round $i\in[k]$ with a constraint:
    \[
            \sum_{j\in[t]} x_{T_j}\cdot\min\left(b^{T_j}_i,q_i\cdot\cost(p)\right)\leq m_i,
    \]
  which ensures that $p$ is not funded instead of $c_i$ in round~$i$ (again, assuming $p$ is last in the tie-breaking order). 
  The only missing component is to secure that $p$ is not selected as the $(k+1)$-th funded project (as $p$ had at least one approval to begin with, such a situation was impossible under $\phragmen$)
  To prevent this, we add a final constraint
  \[
    \sum_{j\in[t]} x_{T_j}\cdot b^{T_j}_{k+1}\leq m_{k+1}.
  \]
  The rest of
  the arguments and the running time remain the same as for
  \phragmen.
\end{proof}

\begin{figure*}[t!]
    \centering 
\begin{subfigure}{0.24\textwidth}
\includegraphics[width=\textwidth]{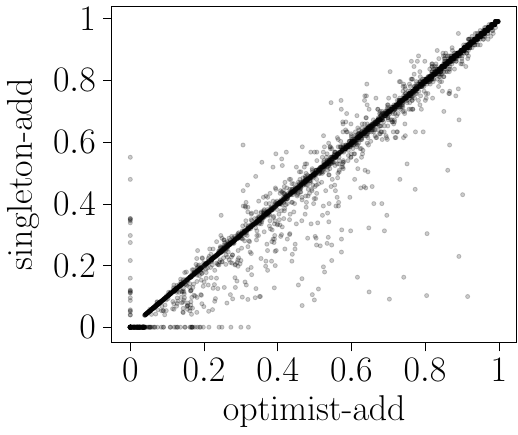}
  \caption{}\label{fig:ov_corrMES2a}
\end{subfigure} \hfill
\begin{subfigure}{0.24\textwidth}
\includegraphics[width=\textwidth]{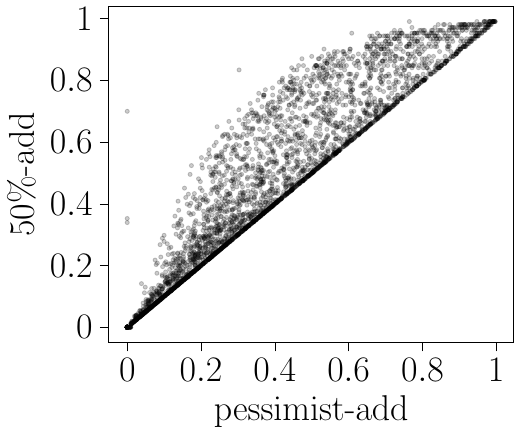}
  \caption{}\label{fig:ov_corrMES2c}
\end{subfigure}\hfill
\begin{subfigure}{0.24\textwidth}
\includegraphics[width=\textwidth]{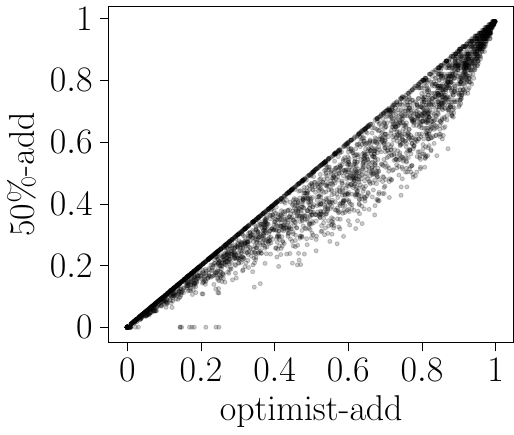}
  \caption{}\label{fig:ov_corrMES2d}
\end{subfigure}\hfill
\begin{subfigure}{0.24\textwidth}
\includegraphics[width=\textwidth]{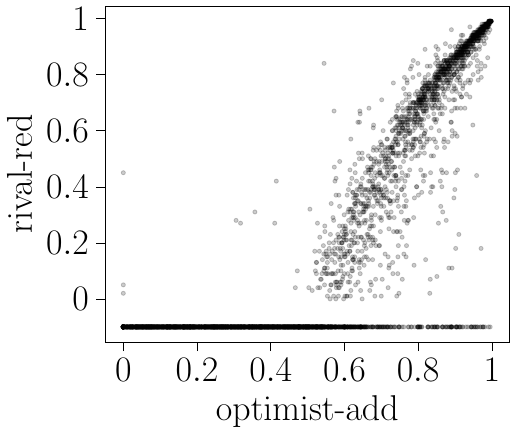}
  \caption{}\label{fig:ov_corrMES2e}
\end{subfigure}
\caption{Further correlation plots between our measures where each point is one project for $\eqphragmen$ (continues \Cref{fig:corr}). Measures are normalized so that 1 denotes no change and 0 denotes a maximal-size change. In \Cref{fig:ov_corrMES2e}, a negative value of $\rivalred$ means that for all considered values of $\ell$ removing rivalry approvals from $\ell$ supporters selected uniformly at random was not sufficient for a $50\%$-funding probability.} 
\label{fig:ov_corrMES2}
\end{figure*} 
\subsection{Adding-Singletons Measure}

\propRandomAddPoly*
\begin{proof}
    In this proof, we assume that the $\singletonadd_E(p)$ measure is well-defined, which can be easily checked in constant time before we run our algorithm; see \Cref{sec:undefined}.
	
    The rough idea for all the rules is very the same. We simulate the execution of a respective PB rule and, in each round, we compute the minimum number of voters approving $p$ only that need to be added to the instance such that~$p$ is funded in this round. The value of $\singletonadd_E(p)$ is then the minimum value over all rounds.
	
    We start with the $\av$ rule. The algorithm is very the same as in \Cref{prop:optimitstadd:poly}. We run the rule on the original instance as usual. The interesting part of the process takes place once the funding of the currently processed project $c$ first decreases the value of the remaining budget below $\cost(p)$. Before funding $c$, we compute $m^*$ as $|A(c)| - |A(p)| + [c \succ p]$, where $[ c \succ p ]$ evaluates to $1$ if $c$ precedes $p$ in tie-breaking order and to $0$ otherwise. The algorithm then outputs the value of $m^*$. For correctness, suppose that $\singletonadd_E^\av(p) = m' < m^*$. If we add $m'$ additional voters approving $p$, then $p$ is assumed by the rule in some round after $c$ is funded. However, $c$ was selected so that, after funding $c$, the remaining budget is below $\cost(p)$. Consequently, adding $m'$ additional voters is not sufficient to fund $p$ and therefore $\singletonadd_E^\av(p) \geq m^*$. Since, by the definition of the rule, the distinguished project $p$ is funded after adding $m^*$ approvals, the algorithm is correct and clearly performs only $O(1)$ additional operations per round.
    
    Next, let the rule be $\ph$. We again simulate the rule and, this time, in every round~$i$, we compute the minimum number of voters who approve solely $p$ that we need to add to the instance to make~$p$ funded in round~$i$. Let $c$ be a project originally funded in round $i$, $b_{i}$ be the budget of a virtual voter approving only~$p$ in this round, and $e_i^p$ be the sum of endowments of all voters who approve $p$ just before round $i$ in the original instance. To make $p$ funded, we need to ensure that $m_i\cdot b_{i} + e_i^p \geq \cost(p)$, where $m_i\in\mathbb{N}$ is the number of added voters. If~${m_i\cdot b_i + e_i^p = \cost(p)}$, the tie-breaking prefers $c$ over $p$, and there exists at least one voter who approves both $p$ and $c$, we need to make the inequality strict.
    We stop the execution once the remaining budget is lower than $\cost(p)$ and the algorithm outputs the minimum~$m_i$ over all rounds.
    
    To conclude, observe that the computation of $m_i$ is $O(1)$ operation along the standard computation of the rule. Moreover, we can create one ``global`` variable $m^*$ to store the minimum $m_i$. The variable is originally initiated to $\infty$, and in each round we conditionally update it if we face a smaller value than the one currently stored. This also requires a constant number of additional operations per round, and the theorem follows. 
\end{proof}

\subsection{Rivalry-Reduction Measure}
\rivalredproof*
\begin{proof}
    Membership to $\np$ is clear for all three rules: given a set of voters $S\subseteq V$ to be adjusted as a certificate, we can verify that this is indeed a solution by checking whether all of them approve $p$ and by setting $A(v) = p$ for every $v\in S$ and running the rule on the modified instance. Since all the rules run in polynomial-time, the membership holds.

    To prove $\np$-hardness under the $\av$ rule, we reduce from the $\setCover$ problem, which is known to be $\np$-complete~\cite{Karp1972}. 
    Let $\mathcal{I}=(U, S, k)$ be an instance of $\setCover$, where $U$ is a set of elements (universe), $S$ is a collection of sets of elements from $U$, and $k$ is an integer. 
    Our goal is to decide whether there exist $k$-sized set $S'\subseteq S$ such that each element of $U$ appears in at least one of the selected sets.

    We construct an equivalent instance $\mathcal{J}$ of our problem as follows. 
    We create one \emph{universe-project} for every element $u\in U$ and add one distinguished project~$p$. All the projects have cost $1$, and we have $|P|=|U|+1$. 
    For each set $S_j \in S$, we create one \emph{set-voter} approving $p$ and candidates corresponding to elements of $S_j$. 
    Furthermore, for each project $u_i \in U$, we create $|S|-|S(u_i)|+1$, where $S(u_i)$ is the set of sets containing $u_i$ as an element, \emph{dummy-voters} who approve only $u_i$. That is, each universe-project is approved by $|S|+1$ voters and $p$ is approved by exactly $|S|$ voters.
    To complete the construction, we set the budget $B=1$, the number of votes to change $\ell = k$, and the tie-breaking order is $p \succ u_1 \succ u_2 \succ \ldots \succ u_{|U|}$. 
    Observe that exactly one project will be funded, as all project are unit-cost and the budget is $B=1$.

    For the correctness, let $\mathcal{I}$ be a yes-instance and $S'\subseteq S$ be a set cover of size~$k$. We remove rivalry approvals of all set-voters corresponding to the sets in $S'$. Since $S'$ is a set cover, every universe project loses at least one approval. Therefore, each universe project is approved by at most $|S|$ voters and so is the distinguished candidate~$p$. Therefore, due to our tie-breaking order, in the modified instance $p$ will be selected as the only funded candidate and, thus, voters corresponding to elements of $S'$ are a solution for $\mathcal{J}$.

\begin{figure*}[t!]
    \centering 
\begin{subfigure}{0.32\textwidth}
\includegraphics[width=\textwidth]{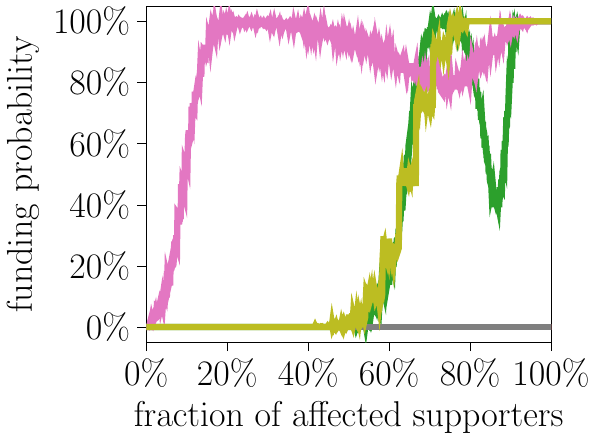}
  \caption{Lodz 2022 (Dolina Lodki)}
\end{subfigure} \hfill
\begin{subfigure}{0.32\textwidth}
\includegraphics[width=\textwidth]{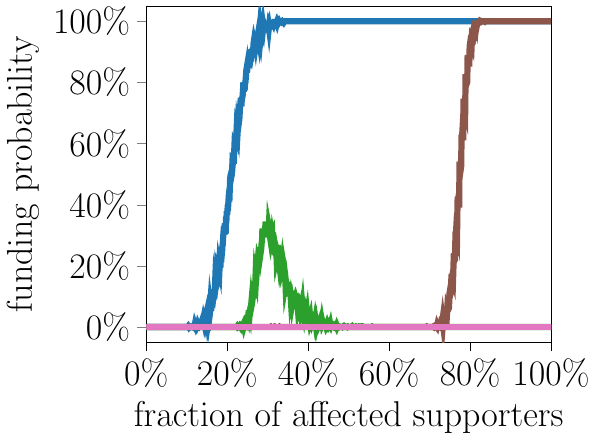}
  \caption{Lodz 2022 (Stare Polesie)}
\end{subfigure} \hfill
\begin{subfigure}{0.32\textwidth}
\includegraphics[width=\textwidth]{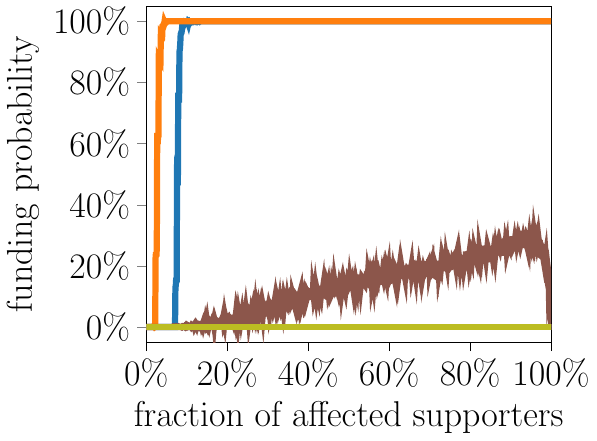}
  \caption{Warszawa 2018 (Nowolipki Powazki)}\label{fig:rv:d_appendix}
\end{subfigure}
\caption{Line plots showing how the funding probability of a project develops if we remove rivalry approvals from 
  its supporters selected uniformly at random for $\eqphragmen$ (continues \Cref{fig:rvex}).} 
\label{fig:rv}
\end{figure*} 

\begin{figure}[t!]
    \centering 
\begin{subfigure}{0.32\textwidth}
\centering
\includegraphics[width=\textwidth]{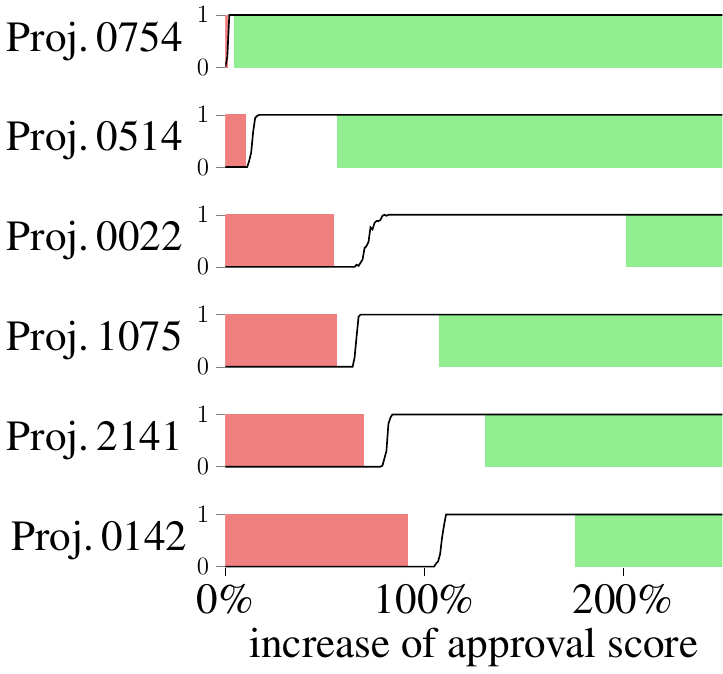}
  \caption{Warszawa 2018 (Wilanow)}
\end{subfigure}\vspace*{0.5cm}
\begin{subfigure}{0.35\textwidth}
\centering
\hspace*{0.2cm}\includegraphics[width=\textwidth]{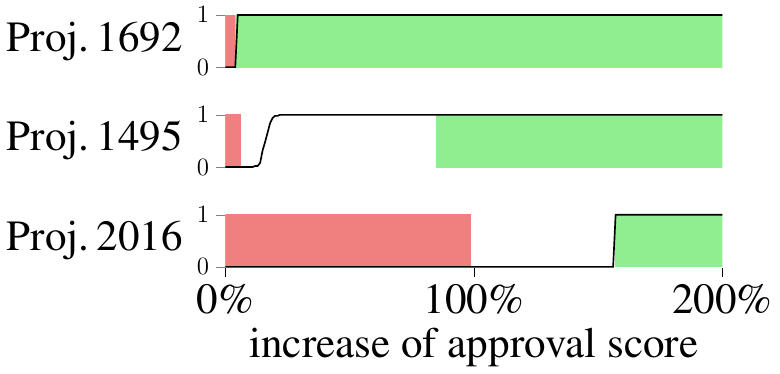}
  \caption{Warszawa 2019 (Ochota)}
\end{subfigure}
\caption{Line plots showing how the funding probability of a project develops from $0$ to $1$ when increasing its approval score by adding approvals uniformly at random to existing voters for $\eqphragmen$ (continues \Cref{fig:50}). The red area goes until the $\optimistadd$ value and the green area extends from the $\pessimistadd$ value.} 
\label{fig:rvA}
\end{figure} 

    In the opposite direction, let $\mathcal{J}$ be a yes-instance and $V'\subseteq V\cap A(p)$ be an $\ell$-sized subset of voters such that setting $A(v)=\{p\}$ for every $v\in V'$ leads to the funding of $p$. Clearly, the set $V'$ consists exclusively of set-voters, as no dummy-voter approves $p$. Since by removing rivalry approvals of voters in $V'$ the distinguished candidate become funded, every universe-project loses at least one approval. In other words, the removed approvals contain each universe-project at least once; setting $S'$ to be the set of all sets in $S$ corresponding to set-voters in $V'$, we obtain $|S'|=|V'|=\ell=k$ and that $S'$ is a set cover in $\mathcal{I}$ by the previous argumentation.

    The reduction can be clearly done in polynomial time and the $\np$-completeness for $\av$ follows. Surprisingly, we can also use the same construction to prove $\np$-completeness for $\ph$ and $\eq$. Observe that there is only one round and all the rules select the most approved candidate as the single winner. Hence, the problem is $\np$-complete for all rules of our interest.
\end{proof}

\section{Additional Material for \Cref{sec:pract}}\label{app:pract}

\subsection{Additional Plots for $\eqphragmen$} \label{app:pract-EQ}

\Cref{fig:ov_corrMES2} shows additional correlation plots for pairs of our measures. 
\Cref{fig:rv} includes further examples of instances where projects show a non-monotonic behavior when removing rivalry approvals from project's supporters uniformly at random. 
Lastly, in \Cref{fig:rvA}, we give some further examples of instances where some projects have a substantial gap between the $\optimistadd$ and $\pessimistadd$ values yet their funding probabilities perform the typical jumps  when adding approvals to existing voters uniformly at random.

\subsection{Analysis for $\phragmen$}\label{app:pract-PH}

In this section, we describe the results of our experimental analysis for $\phragmen$.
The general picture is very similar as for $\eqphragmen$ and our presented figures are created analogous to the ones for $\phragmen$. 
In \Cref{tab:PCC_appendix}, we show the Pearson Correlation Coefficient (PCC) between our measures and in \Cref{fig:ov_corr_app} correlation plots for some of them.
The PCC values for $\phragmen$ are very similar to the ones for $\eqphragmen$, typically differing by at most $0.02$. 
The only larger differences appear for $\rivalred$, which has for $\phragmen$ a higher correlation to the other measures than for $\eqphragmen$. 
(Note that for $\phragmen$ $1484$ out of the $3513$ projects have a funding probability above $50\%$ for some considered value of $\ell$ when removing rivalry approvals from $\ell$ of its supporters uniformly at random). 
Also, the correlation plots for $\phragmen$ look very similar to the ones for $\equalshares$. 
The most significant difference here is probably in \Cref{fig:corr-a} which compares $\singletonadd$ and $\optimistadd$, which is due to the fact that as discussed in the main body for $\phragmen$ $\singletonadd$ constitutes a lower bound for $\optimistadd$, which is not the case for $\eqphragmen$.

Regarding $\randomadd$, we again observe that the measure is slightly closer to $\optimistadd$ than to $\pessimistadd$. 
The average (resp., maximum) difference between $\randomadd$ and  $\optimistadd$ is $0.07$ (resp., $0.3$), whereas it is $0.1$ (resp., $0.4$) for the $\randomadd$ measure and $\pessimistadd$. 
Again we find that even in case there is a gap between the $\optimistadd$ and $\pessimistadd$ value of a project, project's funding probability quickly transitions from an almost $0\%$ to a close-to $100\%$ when adding approvals to existing voters uniformly at random. 
\Cref{fig:50_app} shows the behavior of several instances in support of this claim.

Regarding $\rivalred$, for $\phragmen$ there are more projects whose funding probability behaves non-monotonically when removing rivalry approvals than for $\eqphragmen$. 
Again we present in \Cref{fig:rv_app} some cherry-picked instances where the funding probabilities of some projects show a particularly interesting behavior.

\begin{table}[t]
	\centering
\setlength{\tabcolsep}{3pt}
 \resizebox{0.55\textwidth}{!}{\begin{tabular}{lcccccc}
\toprule
{} &  optimist &  pessimistic &  50\% &  singleton &  rival &  cost  \\
\midrule
optimist      &                   $-$ &                   \corr{0.89} &          \corr{0.98}  &             \corr{0.98} &                       \corr{0.85} &         \corr{0.74}  \\
pessimistic      &                   \corr{0.89} &                   $-$ &          \corr{0.96} &            \corr{0.88}&                       \corr{0.76} &         \corr{0.79}  \\
50\%             &                   \corr{0.98} &                   \corr{0.96} &          $-$ &             \corr{0.96}&                       \corr{0.85} &         \corr{0.79}  \\
singleton          &                  \corr{0.98} &                   \corr{0.88} &          \corr{0.96} &             $-$ &                       \corr{0.93} &         \corr{0.78}  \\
rival &                   \corr{0.85} &                  \corr{0.76} &          \corr{0.85} &             \corr{0.93} &                       $-$ &         \corr{0.70}  \\
cost               &                   \corr{0.74} &                   \corr{0.79} &          \corr{0.79} &             \corr{0.78} &                       \corr{0.70} &         $-$  \\
\bottomrule
\end{tabular}}
\caption{Pearson Correlation Coefficient between measures for $\phragmen$.} 
\label{tab:PCC_appendix}
\end{table}

\begin{figure*}[t!]
    \centering 
\begin{subfigure}{0.32\textwidth}
\includegraphics[width=\textwidth]{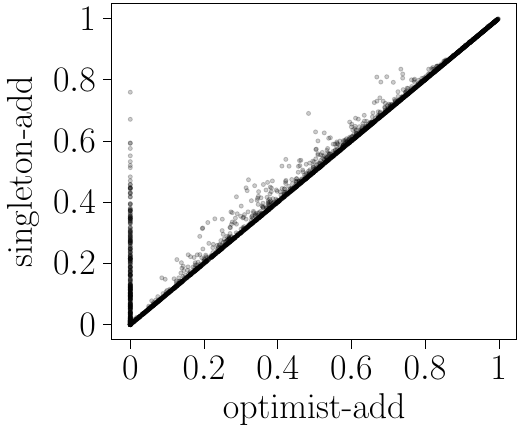}
  \caption{}\label{fig:corr-a}
\end{subfigure} \hfill
\begin{subfigure}{0.32\textwidth}
\includegraphics[width=\textwidth]{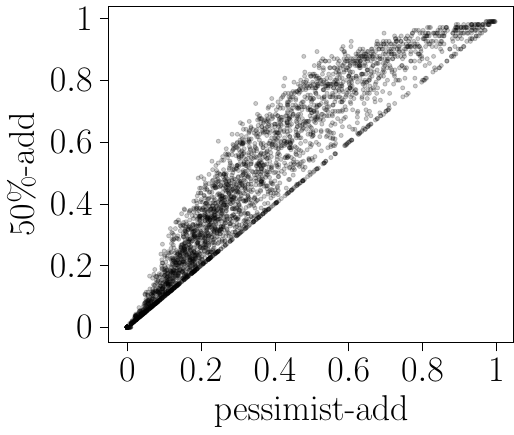}
  \caption{}\label{fig:c}
\end{subfigure} \hfill
\begin{subfigure}{0.32\textwidth}
\includegraphics[width=\textwidth]{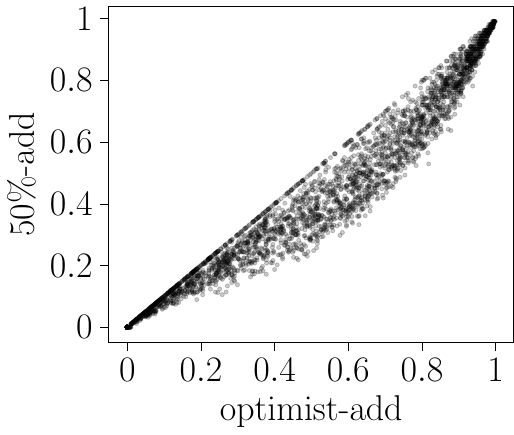}
  \caption{}\label{fig:d}
\end{subfigure}\hfill
\begin{subfigure}{0.32\textwidth}
\includegraphics[width=\textwidth]{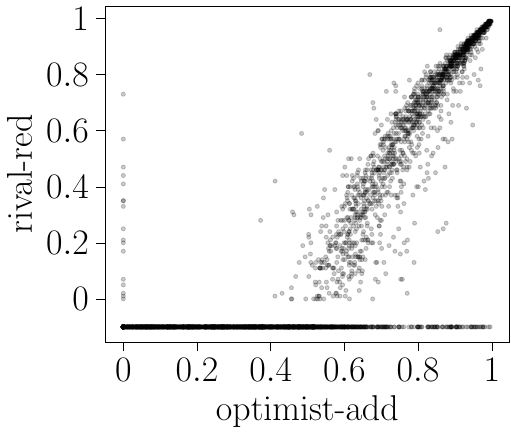}
  \caption{}\label{fig:e}
\end{subfigure}\hfill
\begin{subfigure}{0.32\textwidth}
\includegraphics[width=\textwidth]{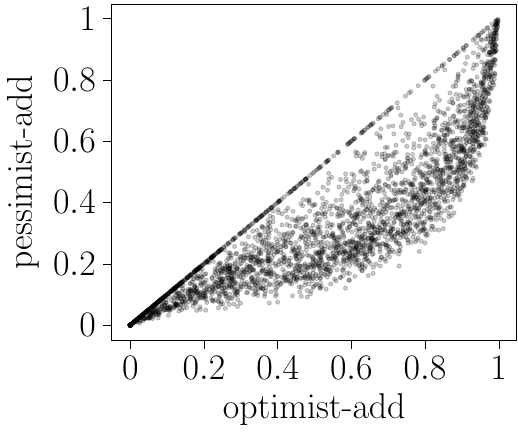}
  \caption{}\label{fig:b}
\end{subfigure}\hfill
\begin{subfigure}{0.32\textwidth}
\includegraphics[width=\textwidth]{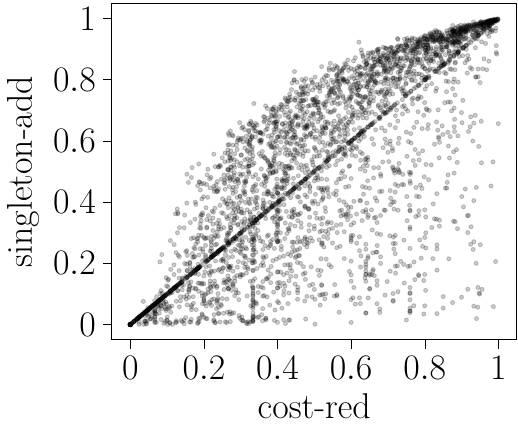}
  \caption{}\label{fig:ov_corr_appf}
\end{subfigure}
\caption{Correlation plots where each point is one project for $\phragmen$. A negative value of $\rivalred$ means that for all considered values of $\ell$ removing rivalry approvals from $\ell$ supporters selected uniformly at random was not sufficient for a $50\%$-funding probability.} 
\label{fig:ov_corr_app}
\end{figure*} 

\begin{figure*}[t!]
    \centering 
    \begin{subfigure}{0.32\textwidth}
\includegraphics[width=\textwidth]{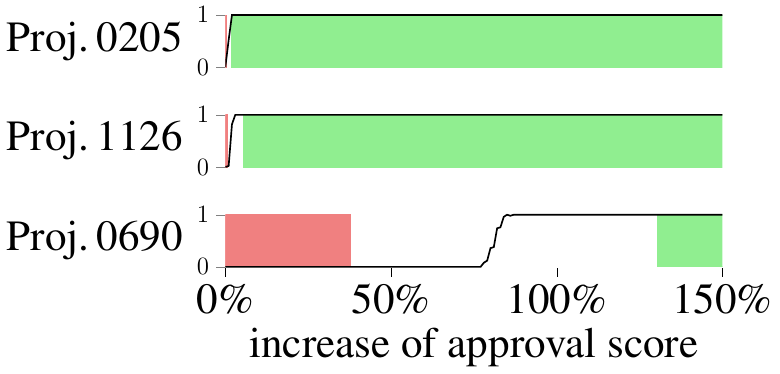}
  \caption{Warszawa 2018 (Lotnisko Bemowo Lotnisko Fort Bema)}\label{fig:rv:e}
\end{subfigure}\hfill
\begin{subfigure}{0.32\textwidth}
\includegraphics[width=\textwidth]{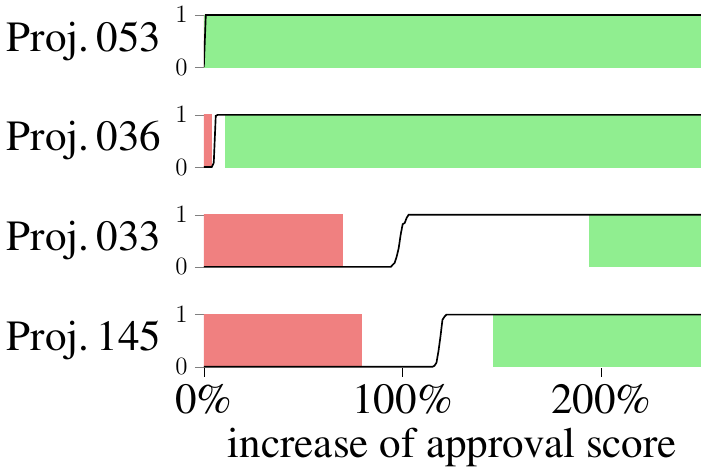}
  \caption{Lodz 2022 (Zlotno)}
\end{subfigure} \hfill
\begin{subfigure}{0.32\textwidth}
\includegraphics[width=\textwidth]{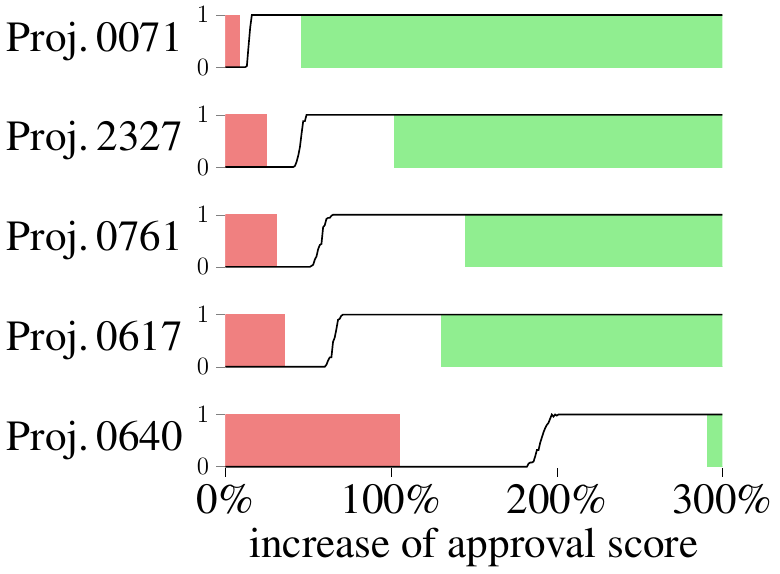}
  \caption{Warszawa 2017 (Szczesliwice)}
\end{subfigure} \hfill
\begin{subfigure}{0.32\textwidth}
\includegraphics[width=\textwidth]{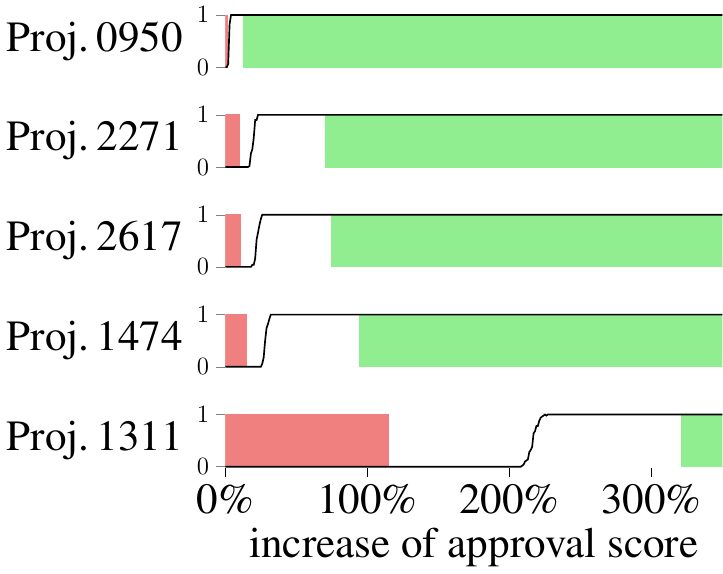}
  \caption{Warszawa 2017 (Wars)}
\end{subfigure}\hfill
\begin{subfigure}{0.32\textwidth}
\includegraphics[width=\textwidth]{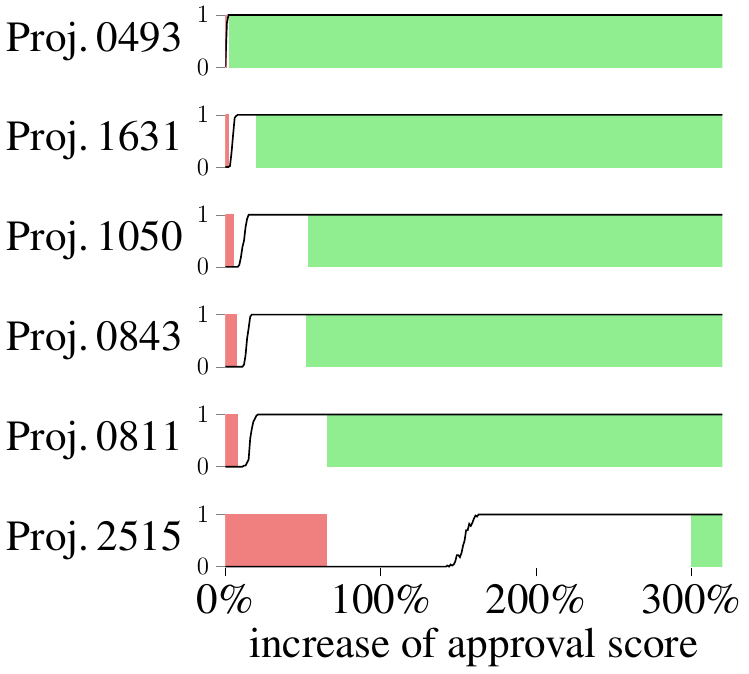}
  \caption{Warszawa 2018 (Brodno)}
\end{subfigure}\hfill
\begin{subfigure}{0.32\textwidth}
\includegraphics[width=\textwidth]{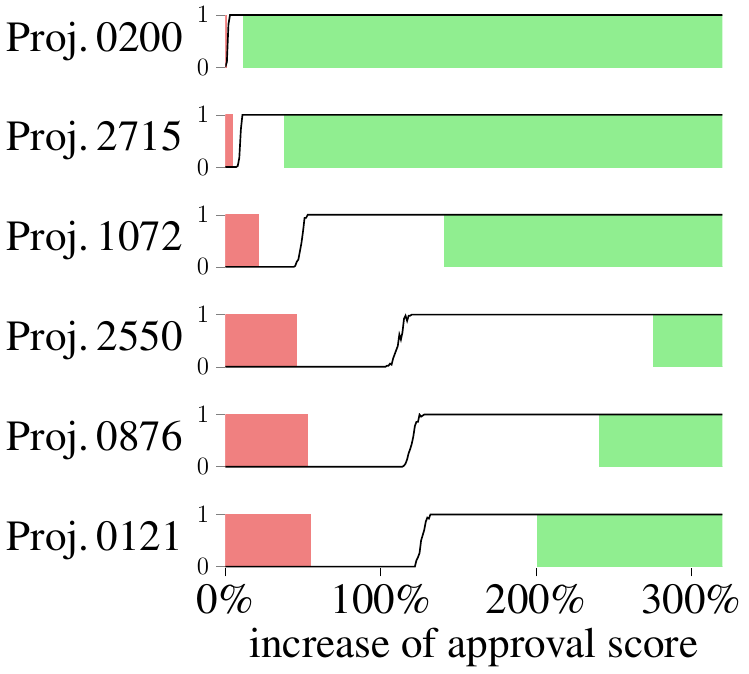}
  \caption{Warszawa 2018 (Goclaw)}\end{subfigure}
\caption{For $\phragmen$, line plots showing how the funding probability of projects develops from $0$ to $1$ when increasing their approval scores by adding approvals uniformly at random to existing voters. The red area goes until the $\optimistadd$ value and the green area extends from the $\pessimistadd$ value.} 
\label{fig:50_app}
\end{figure*} 

\begin{figure*}[t!]
    \centering 
\begin{subfigure}{0.32\textwidth}
\includegraphics[width=\textwidth]{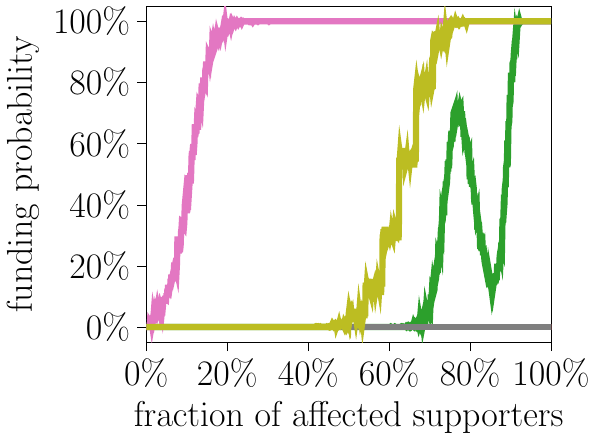}
  \caption{Lodz 2022 (Dolina Lodki)}
\end{subfigure} \hfill
\begin{subfigure}{0.32\textwidth}
\includegraphics[width=\textwidth]{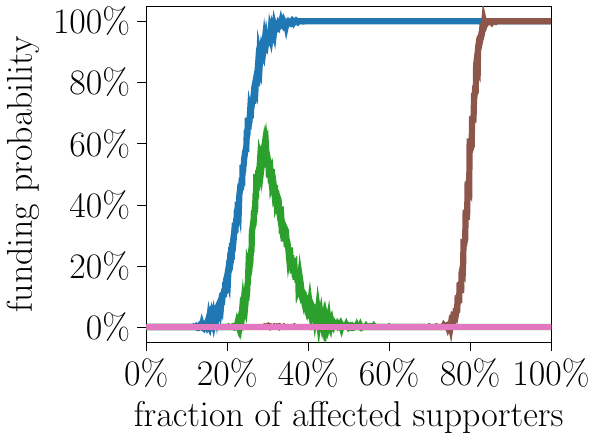}
  \caption{Lodz 2022 (Stare Polesie)}
\end{subfigure} \hfill
\begin{subfigure}{0.32\textwidth}
\includegraphics[width=\textwidth]{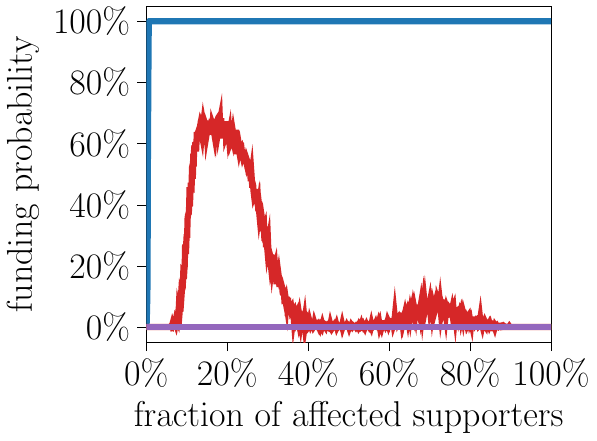}
  \caption{Lodz 2022 (Zlotno)}
\end{subfigure}\hfill
\begin{subfigure}{0.32\textwidth}
\includegraphics[width=\textwidth]{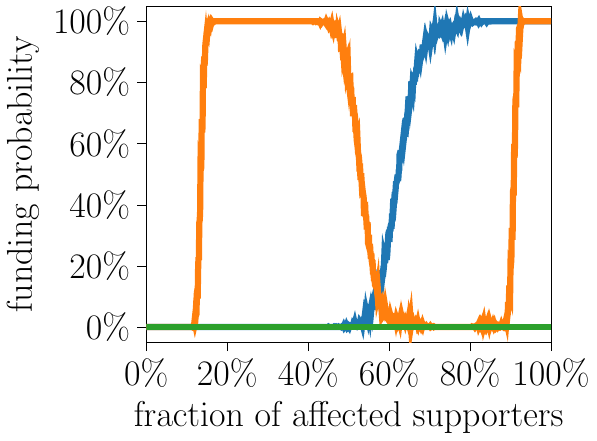}
  \caption{Warszawa 2018 (Brodno Podgrodzie)}
\end{subfigure}\hfill
\begin{subfigure}{0.32\textwidth}
\includegraphics[width=\textwidth]{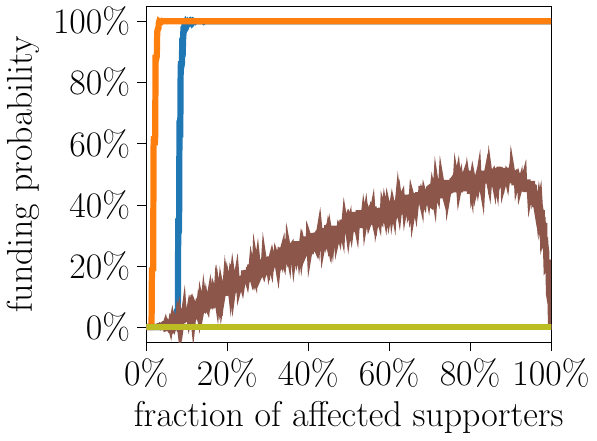}
  \caption{Warszawa 2018 (Nowolipki Powazki)}\end{subfigure}\hfill
\begin{subfigure}{0.32\textwidth}
\includegraphics[width=\textwidth]{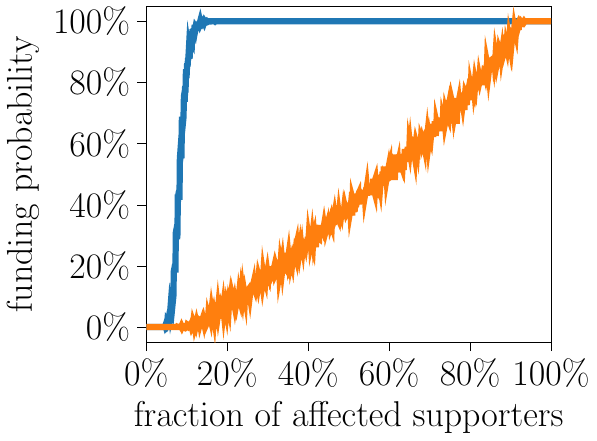}
  \caption{Warszawa 2018 (Miedzeszyn)}
\end{subfigure}
\caption{For $\phragmen$, line plots showing how the funding probability of projects develops if we remove rivalry approvals from subgroups of their supporters selected uniformly at random.} 
\label{fig:rv_app}
\end{figure*} 

\subsection{Analysis for $\greedyav$}\label{app:pract-AV}

We repeated our experiments for $\greedyav$ where $3581$ projects ended up being not funded. 
Notably, here all measures involving adding approvals to existing voters or adding singleton voters coincide, so we only analyze $\costred$, $\singletonadd$, and $\rivalred$. 

\Cref{fig:ov_corrGreedyAV} shows correlation plots and PCC values for these three pairs of measures. 
Regarding the relationship between $\singletonadd$ and $\costred$ (\Cref{fig:ov_corrGreedyAV1}), the connection between the two is much weaker than for the proportional rules. 
In particular, there are many projects with a small value of $\costred$ and varying values of $\singletonadd$. 
This can be explained quite easily by recalling the inner workings of $\greedyav$: When modifying only project's costs, the ordering in which projects are considered remains unchanged. 
Thus, a project can only cost the amount of money left when it is their turn in the original ordering of projects. 
Now, in case most or all of the budget has already been spent before, there is no money for the project left. 
In contrast, modifying the approval score of a project allows us to influence its position in the order in which projects are considered and thereby how much money is left for the project once it is its turn. 
Accordingly, the difference to the proportional rules here is due to the fact that for $\greedyav$, $\costred$ is in some sense of more limited power. 

For $\rivalred$, for $862$ projects for some considered value of $\ell$ removing rivalry approvals from $\ell$ supporters selected uniformly at random was sufficient for a $50\%$ funding probability.
Notably, in contrast to the proportional rules, removing rivalry approvals has a different, arguably weaker effect for $\greedyav$, as it only reduces the approval score of competing projects and not how much of their preallocated money voters have left to spend on the designated project (because there is no preallocated money for $\greedyav$). 
Thus, it is to be expected that $\rivalred$ is less powerful for $\greedyav$ than for the proportional rules.
In fact, the fact that $\rivalred$ has an impact on $862$ projects can be even regarded as a surprisingly high number, highlighting that also for $\greedyav$ adding additional approvals to one's ballot can hurt the funding possibilities of other approved projects. 

Remarkably, for $\greedyav$, there are even more projects than for the proportional rules whose funding probability behaves non-monotonically when removing rivalry approvals and the behavior of these projects comes in more different flavors. 
\Cref{fig:50_app_AV} shows some selected instances. 
One instance that sticks out in particular is \Cref{fig:Goclaw}, where the gray project has an almost $100\%$ funding probability when few of its supporters remove their rivalry approvals; however, in case some more do it as well its funding probability drops again to $0\%$ and even in case gray's supporters don't approve any other projects, the project does not get funded again. 
Nevertheless, for a majority of projects, their funding probability does behave monotonically and quickly jumps from around $0\%$ to around $100\%$ (the instances shown in \Cref{fig:50_app_AV}  are meant to demonstrate the non-monotonic behavior and not to provide representative coverage of the project's behavior). 
Examining the relationship of $\rivalred$ to the other two measures in more detail, we see almost no positive correlation with $\costred$ in \Cref{fig:ov_corrGreedyAV2}. 
In contrasts, the connection to $\singletonadd$ (see \Cref{fig:ov_corrGreedyAV3}) is stronger. 
The correlation (plot) of $\singletonadd$ and $\rivalred$ is quite similar to the respective plots for $\phragmen$ (\Cref{fig:ov_corr_appf}) and $\eqphragmen$ (\Cref{fig:ov_corrMES2e}).  
The above observations indicate that also the power of $\rivalred$ for $\phragmen$ and $\eqphragmen$ stems only partly from the proportionality of the two rules, yet also from the fact that there is simply more money left in case competing projects are not funded. 

\begin{figure*}[t!]
    \centering 
\begin{subfigure}{0.3\textwidth}
\includegraphics[width=\textwidth]{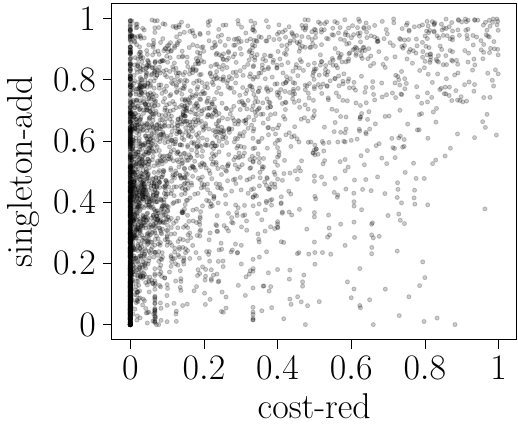}
  \caption{PCC $0.47$}\label{fig:ov_corrGreedyAV1}
\end{subfigure} \hfill
\begin{subfigure}{0.3\textwidth}
\includegraphics[width=\textwidth]{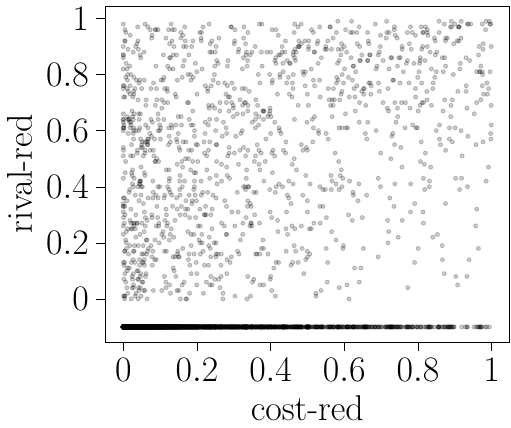}
  \caption{PCC $0.33$}\label{fig:ov_corrGreedyAV2}
\end{subfigure}\hfill
\begin{subfigure}{0.3\textwidth}
\includegraphics[width=\textwidth]{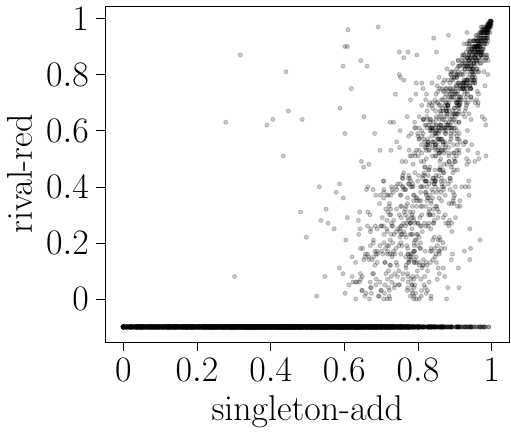}
  \caption{PCC $0.68$}\label{fig:ov_corrGreedyAV3}
\end{subfigure}
\caption{For $\greedyav$, correlation plots where each point is one project. A negative value of $\rivalred$ means that for all considered values of $\ell$ removing rivalry approvals from $\ell$ supporters selected uniformly at random was not sufficient for a $50\%$-funding probability.} 
\label{fig:ov_corrGreedyAV}
\end{figure*} 

\begin{figure*}[t!]
    \centering 
\begin{subfigure}[t]{0.24\textwidth}
\includegraphics[width=\textwidth]{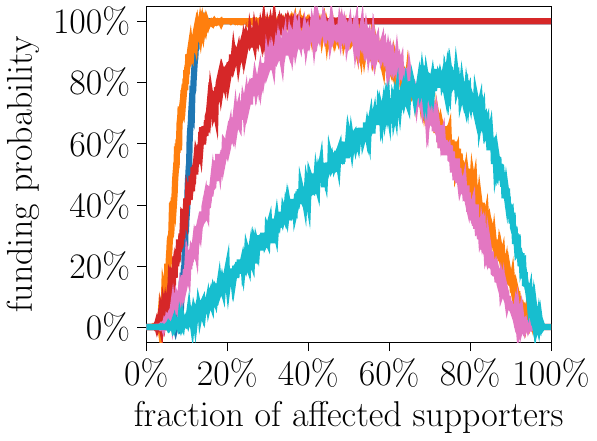}
  \caption{Lodz 2020 (Baluty Doly)}
\end{subfigure} \hfill
\begin{subfigure}[t]{0.24\textwidth}
\includegraphics[width=\textwidth]{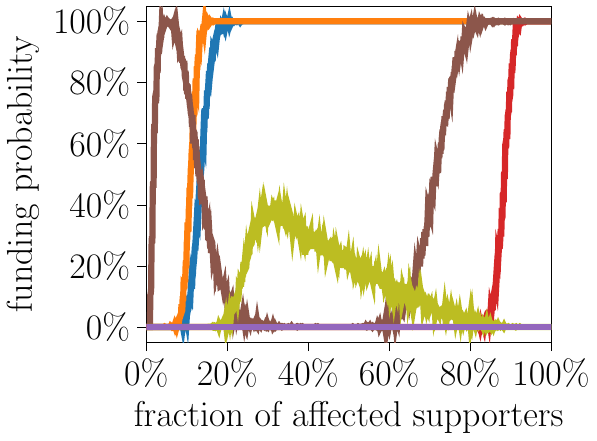}
  \caption{Warszawa 2017 (Choszczowka Dabrowka Szlachecka Bialoleka Dworska Henrykow Szamocin)}
\end{subfigure} \hfill
\begin{subfigure}[t]{0.24\textwidth}
\includegraphics[width=\textwidth]{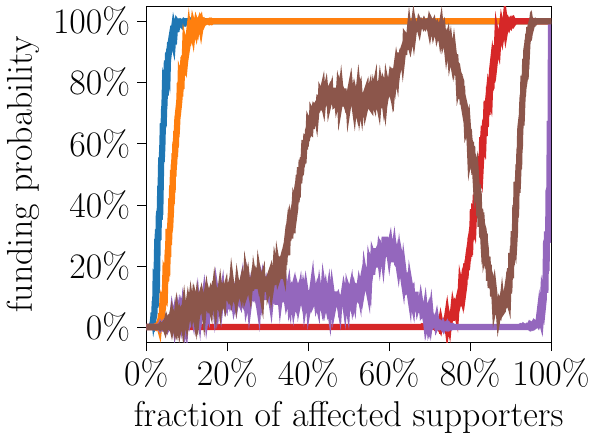}
  \caption{Warszawa 2017 (Saska Kepa)}
\end{subfigure}\hfill
\begin{subfigure}[t]{0.24\textwidth}
\includegraphics[width=\textwidth]{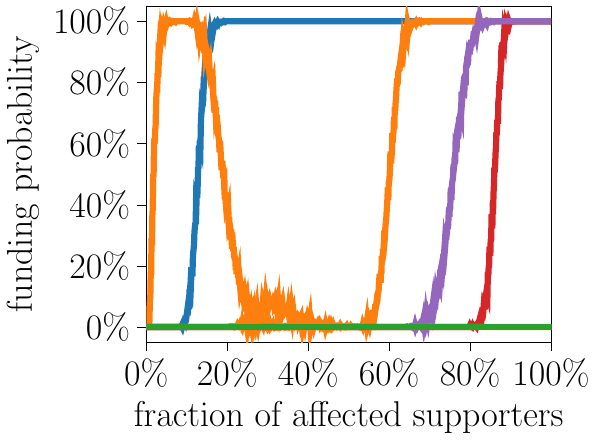}
  \caption{Warszawa 2018 (Bialoleka Obszar 2)}
\end{subfigure}\hfill
\begin{subfigure}[t]{0.24\textwidth}
\includegraphics[width=\textwidth]{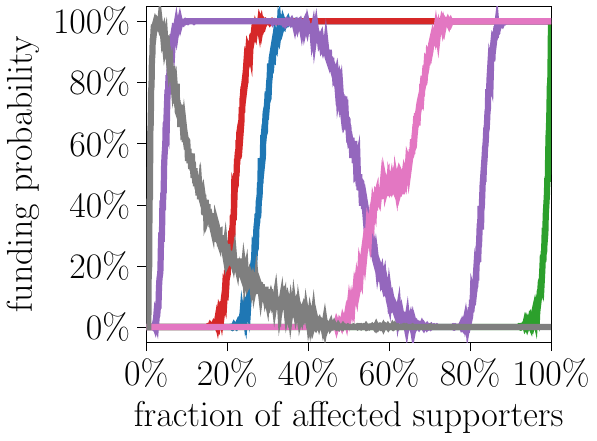}
  \caption{Warszawa 2018 (Goclaw)} \label{fig:Goclaw}\end{subfigure}\hfill
\begin{subfigure}[t]{0.24\textwidth}
\includegraphics[width=\textwidth]{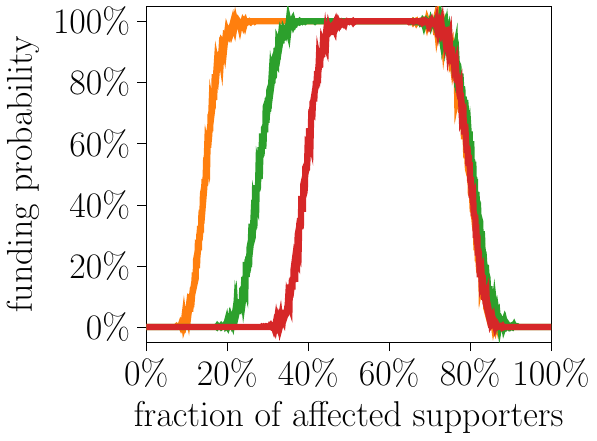}
  \caption{Warszawa 2018 (Grochow Centrum)}
\end{subfigure}
\hfill
\begin{subfigure}[t]{0.24\textwidth}
\includegraphics[width=\textwidth]{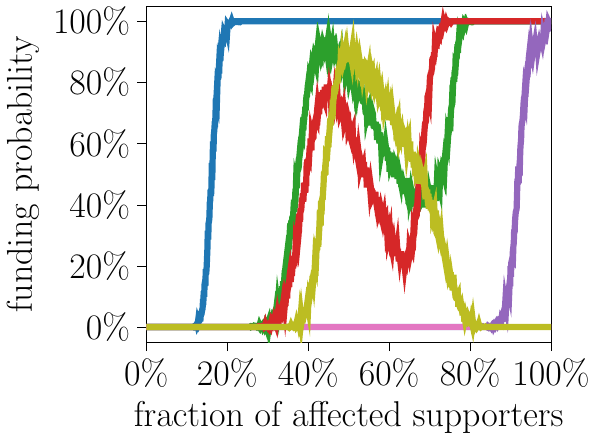}
  \caption{Warszawa 2018 (Szczesliwice)}
\end{subfigure}
\hfill
\begin{subfigure}[t]{0.24\textwidth}
\includegraphics[width=\textwidth]{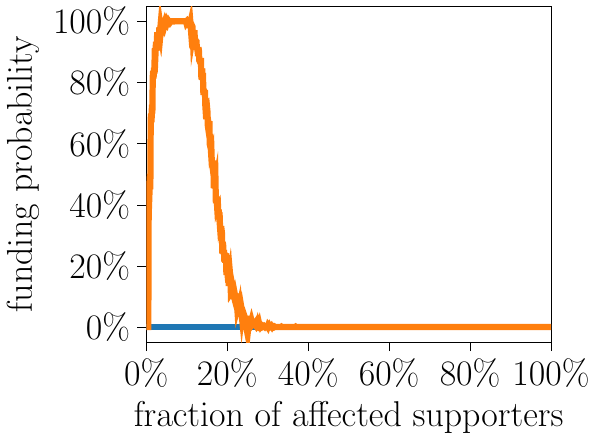}
  \caption{Warszawa 2019 (Obszar 2 Sady Zoliborskie Zatrasie Rudawka)}
\end{subfigure}
\caption{For $\greedyav$, line plots showing how the funding probability of projects develops if we remove rivalry approvals from subgroups of their supporters selected uniformly at random.} 
\label{fig:50_app_AV}
\end{figure*} 

\clearpage
\clearpage

\end{document}

%% file: plots_MES/add_single.tex
\begin{tikzpicture}

\definecolor{darkgray176}{RGB}{176,176,176}
\definecolor{darkorange25512714}{RGB}{255,127,14}
\definecolor{lightgray204}{RGB}{204,204,204}
\definecolor{steelblue31119180}{RGB}{31,119,180}

\begin{axis}[
unit vector ratio*=1 0.3 1,
width=8cm,
legend cell align={left},
legend columns=3,
legend style={
  fill opacity=0.8,
  draw opacity=1,
  text opacity=1,
  at={(0.5,1.25)},
  anchor=north,
  draw=lightgray204
},
tick align=outside,
tick pos=left,
x grid style={darkgray176},
xlabel={increase of approval score by adding singletons},
xmin=0, xmax=100,
xtick style={color=black},
y grid style={darkgray176},
ymin=-5, ymax=105,
ytick style={color=black},
xticklabels={0\%,20\%,40\%,60\%,80\%,100\%},xtick={0,20,40,60,80,100},
yticklabels={not funded, funded},ytick={0,100}
]
\addplot [line width=3pt, steelblue31119180]
table {%
	0 0
	0.1 0
	0.2 0
	0.3 0
	0.4 0
	0.5 0
	0.6 0
	0.7 0
	0.8 0
	0.9 0
	1 0
	1.1 0
	1.2 0
	1.3 0
	1.4 0
	1.5 0
	1.6 0
	1.7 0
	1.8 0
	1.9 0
	2 0
	2.1 0
	2.2 0
	2.3 0
	2.4 0
	2.5 0
	2.6 0
	2.7 0
	2.8 0
	2.9 0
	3 0
	3.1 0
	3.2 0
	3.3 0
	3.4 0
	3.5 0
	3.6 0
	3.7 0
	3.8 0
	3.9 0
	4 0
	4.1 0
	4.2 0
	4.3 0
	4.4 0
	4.5 0
	4.6 0
	4.7 0
	4.8 0
	4.9 0
	5 0
	5.1 0
	5.2 0
	5.3 0
	5.4 0
	5.5 0
	5.6 0
	5.7 0
	5.8 0
	5.9 0
	6 0
	6.1 0
	6.2 0
	6.3 0
	6.4 0
	6.5 0
	6.6 0
	6.7 0
	6.8 0
	6.9 0
	7 0
	7.1 0
	7.2 0
	7.3 0
	7.4 0
	7.5 0
	7.6 0
	7.7 0
	7.8 0
	7.9 0
	8 0
	8.1 0
	8.2 0
	8.3 0
	8.4 0
	8.5 0
	8.6 0
	8.7 0
	8.8 0
	8.9 0
	9 0
	9.1 0
	9.2 0
	9.3 0
	9.4 0
	9.5 0
	9.6 0
	9.7 0
	9.8 0
	9.9 0
	10 0
	10.1 0
	10.2 0
	10.3 0
	10.4 0
	10.5 0
	10.6 0
	10.7 0
	10.8 0
	10.9 0
	11 0
	11.1 0
	11.2 0
	11.3 0
	11.4 0
	11.5 0
	11.6 0
	11.7 0
	11.8 0
	11.9 0
	12 0
	12.1 0
	12.2 0
	12.3 0
	12.4 0
	12.5 0
	12.6 0
	12.7 0
	12.8 0
	12.9 0
	13 0
	13.1 0
	13.2 0
	13.3 0
	13.4 0
	13.5 0
	13.6 0
	13.7 0
	13.8 0
	13.9 0
	14 0
	14.1 0
	14.2 0
	14.3 0
	14.4 0
	14.5 0
	14.6 0
	14.7 0
	14.8 0
	14.9 0
	15 0
	15.1 0
	15.2 0
	15.3 0
	15.4 0
	15.5 0
	15.6 0
	15.7 0
	15.8 0
	15.9 0
	16 0
	16.1 0
	16.2 0
	16.3 0
	16.4 0
	16.5 0
	16.6 0
	16.7 0
	16.8 0
	16.9 0
	17 0
	17.1 0
	17.2 0
	17.3 0
	17.4 0
	17.5 0
	17.6 0
	17.7 0
	17.8 0
	17.9 0
	18 0
	18.1 0
	18.2 0
	18.3 0
	18.4 0
	18.5 0
	18.6 0
	18.7 0
	18.8 0
	18.9 0
	19 0
	19.1 0
	19.2 0
	19.3 0
	19.4 0
	19.5 0
	19.6 0
	19.7 0
	19.8 0
	19.9 0
	20 0
	20.1 0
	20.2 0
	20.3 0
	20.4 0
	20.5 0
	20.6 0
	20.7 0
	20.8 0
	20.9 0
	21 0
	21.1 0
	21.2 0
	21.3 0
	21.4 0
	21.5 0
	21.6 0
	21.7 0
	21.8 0
	21.9 0
	22 0
	22.1 0
	22.2 0
	22.3 0
	22.4 0
	22.5 0
	22.6 0
	22.7 0
	22.8 0
	22.9 0
	23 0
	23.1 0
	23.2 0
	23.3 0
	23.4 0
	23.5 0
	23.6 0
	23.7 0
	23.8 0
	23.9 0
	24 0
	24.1 0
	24.2 0
	24.3 0
	24.4 0
	24.5 0
	24.6 0
	24.7 0
	24.8 0
	24.9 0
	25 0
	25.1 0
	25.2 0
	25.3 0
	25.4 0
	25.5 0
	25.6 0
	25.7 0
	25.8 0
	25.9 0
	26 0
	26.1 100
	26.2 100
	26.3 100
	26.4 100
	26.5 100
	26.6 100
	26.7 100
	26.8 100
	26.9 100
	27 100
	27.1 100
	27.2 100
	27.3 100
	27.4 100
	27.5 100
	27.6 100
	27.7 100
	27.8 100
	27.9 100
	28 100
	28.1 100
	28.2 100
	28.3 100
	28.4 100
	28.5 100
	28.6 100
	28.7 100
	28.8 100
	28.9 100
	29 100
	29.1 100
	29.2 100
	29.3 100
	29.4 100
	29.5 100
	29.6 100
	29.7 100
	29.8 100
	29.9 100
	30 100
	30.1 100
	30.2 100
	30.3 100
	30.4 100
	30.5 100
	30.6 100
	30.7 100
	30.8 100
	30.9 100
	31 100
	31.1 100
	31.2 100
	31.3 100
	31.4 0
	31.5 0
	31.6 0
	31.7 0
	31.8 0
	31.9 0
	32 0
	32.1 0
	32.2 0
	32.3 0
	32.4 0
	32.5 0
	32.6 0
	32.7 0
	32.8 0
	32.9 0
	33 0
	33.1 0
	33.2 0
	33.3 0
	33.4 0
	33.5 0
	33.6 0
	33.7 0
	33.8 0
	33.9 0
	34 0
	34.1 0
	34.2 0
	34.3 0
	34.4 0
	34.5 0
	34.6 0
	34.7 0
	34.8 0
	34.9 0
	35 0
	35.1 0
	35.2 0
	35.3 0
	35.4 0
	35.5 0
	35.6 0
	35.7 0
	35.8 0
	35.9 0
	36 0
	36.1 0
	36.2 0
	36.3 0
	36.4 0
	36.5 0
	36.6 0
	36.7 0
	36.8 0
	36.9 0
	37 0
	37.1 0
	37.2 0
	37.3 0
	37.4 0
	37.5 0
	37.6 0
	37.7 0
	37.8 0
	37.9 0
	38 0
	38.1 100
	38.2 100
	38.3 100
	38.4 100
	38.5 100
	38.6 100
	38.7 100
	38.8 100
	38.9 100
	39 100
	39.1 100
	39.2 100
	39.3 100
	39.4 100
	39.5 100
	39.6 100
	39.7 100
	39.8 100
	39.9 100
	40 100
	40.1 100
	40.2 100
	40.3 100
	40.4 100
	40.5 100
	40.6 100
	40.7 100
	40.8 100
	40.9 100
	41 100
	41.1 100
	41.2 100
	41.3 100
	41.4 100
	41.5 100
	41.6 100
	41.7 100
	41.8 100
	41.9 100
	42 100
	42.1 100
	42.2 100
	42.3 100
	42.4 100
	42.5 100
	42.6 100
	42.7 100
	42.8 100
	42.9 100
	43 100
	43.1 100
	43.2 100
	43.3 100
	43.4 100
	43.5 100
	43.6 100
	43.7 100
	43.8 100
	43.9 100
	44 100
	44.1 100
	44.2 100
	44.3 100
	44.4 100
	44.5 100
	44.6 100
	44.7 100
	44.8 100
	44.9 100
	45 100
	45.1 100
	45.2 100
	45.3 100
	45.4 100
	45.5 100
	45.6 100
	45.7 100
	45.8 100
	45.9 100
	46 100
	46.1 100
	46.2 100
	46.3 100
	46.4 100
	46.5 100
	46.6 100
	46.7 100
	46.8 100
	46.9 100
	47 100
	47.1 100
	47.2 100
	47.3 100
	47.4 100
	47.5 100
	47.6 100
	47.7 100
	47.8 100
	47.9 100
	48 100
	48.1 100
	48.2 100
	48.3 100
	48.4 100
	48.5 100
	48.6 100
	48.7 100
	48.8 100
	48.9 100
	49 100
	49.1 100
	49.2 100
	49.3 100
	49.4 100
	49.5 100
	49.6 100
	49.7 100
	49.8 100
	49.9 100
	50 100
	50.1 100
	50.2 100
	50.3 100
	50.4 100
	50.5 100
	50.6 100
	50.7 100
	50.8 100
	50.9 100
	51 100
	51.1 100
	51.2 100
	51.3 100
	51.4 100
	51.5 100
	51.6 100
	51.7 100
	51.8 100
	51.9 100
	52 100
	52.1 100
	52.2 100
	52.3 100
	52.4 100
	52.5 100
	52.6 100
	52.7 100
	52.8 100
	52.9 100
	53 100
	53.1 100
	53.2 100
	53.3 100
	53.4 100
	53.5 100
	53.6 100
	53.7 100
	53.8 100
	53.9 100
	54 100
	54.1 100
	54.2 100
	54.3 100
	54.4 100
	54.5 100
	54.6 100
	54.7 100
	54.8 100
	54.9 100
	55 100
	55.1 100
	55.2 100
	55.3 100
	55.4 100
	55.5 100
	55.6 100
	55.7 100
	55.8 100
	55.9 100
	56 100
	56.1 100
	56.2 100
	56.3 100
	56.4 100
	56.5 100
	56.6 100
	56.7 100
	56.8 100
	56.9 100
	57 100
	57.1 100
	57.2 100
	57.3 100
	57.4 100
	57.5 100
	57.6 100
	57.7 100
	57.8 100
	57.9 100
	58 100
	58.1 100
	58.2 100
	58.3 100
	58.4 100
	58.5 100
	58.6 100
	58.7 100
	58.8 100
	58.9 100
	59 100
	59.1 100
	59.2 100
	59.3 100
	59.4 100
	59.5 100
	59.6 100
	59.7 100
	59.8 100
	59.9 100
	60 100
	60.1 100
	60.2 100
	60.3 100
	60.4 100
	60.5 100
	60.6 100
	60.7 100
	60.8 100
	60.9 100
	61 100
	61.1 100
	61.2 100
	61.3 100
	61.4 100
	61.5 100
	61.6 100
	61.7 100
	61.8 100
	61.9 100
	62 100
	62.1 100
	62.2 100
	62.3 100
	62.4 100
	62.5 100
	62.6 100
	62.7 100
	62.8 100
	62.9 100
	63 100
	63.1 100
	63.2 100
	63.3 100
	63.4 100
	63.5 100
	63.6 100
	63.7 100
	63.8 100
	63.9 100
	64 100
	64.1 100
	64.2 100
	64.3 100
	64.4 100
	64.5 100
	64.6 100
	64.7 100
	64.8 100
	64.9 100
	65 100
	65.1 100
	65.2 100
	65.3 100
	65.4 100
	65.5 100
	65.6 100
	65.7 100
	65.8 100
	65.9 100
	66 100
	66.1 100
	66.2 100
	66.3 100
	66.4 100
	66.5 100
	66.6 100
	66.7 100
	66.8 100
	66.9 100
	67 100
	67.1 100
	67.2 100
	67.3 100
	67.4 100
	67.5 100
	67.6 100
	67.7 100
	67.8 100
	67.9 100
	68 100
	68.1 100
	68.2 100
	68.3 100
	68.4 100
	68.5 100
	68.6 100
	68.7 100
	68.8 100
	68.9 100
	69 100
	69.1 100
	69.2 100
	69.3 100
	69.4 100
	69.5 100
	69.6 100
	69.7 100
	69.8 100
	69.9 100
	70 100
	70.1 100
	70.2 100
	70.3 100
	70.4 100
	70.5 100
	70.6 100
	70.7 100
	70.8 100
	70.9 100
	71 100
	71.1 100
	71.2 100
	71.3 100
	71.4 100
	71.5 100
	71.6 100
	71.7 100
	71.8 100
	71.9 100
	72 100
	72.1 100
	72.2 100
	72.3 100
	72.4 100
	72.5 100
	72.6 100
	72.7 100
	72.8 100
	72.9 100
	73 100
	73.1 100
	73.2 100
	73.3 100
	73.4 100
	73.5 100
	73.6 100
	73.7 100
	73.8 100
	73.9 100
	74 100
	74.1 100
	74.2 100
	74.3 100
	74.4 100
	74.5 100
	74.6 100
	74.7 100
	74.8 100
	74.9 100
	75 100
	75.1 100
	75.2 100
	75.3 100
	75.4 100
	75.5 100
	75.6 100
	75.7 100
	75.8 100
	75.9 100
	76 100
	76.1 100
	76.2 100
	76.3 100
	76.4 100
	76.5 100
	76.6 100
	76.7 100
	76.8 100
	76.9 100
	77 100
	77.1 100
	77.2 100
	77.3 100
	77.4 100
	77.5 100
	77.6 100
	77.7 100
	77.8 100
	77.9 100
	78 100
	78.1 100
	78.2 100
	78.3 100
	78.4 100
	78.5 100
	78.6 100
	78.7 100
	78.8 100
	78.9 100
	79 100
	79.1 100
	79.2 100
	79.3 100
	79.4 100
	79.5 100
	79.6 100
	79.7 100
	79.8 100
	79.9 100
	80 100
	80.1 100
	80.2 100
	80.3 100
	80.4 100
	80.5 100
	80.6 100
	80.7 100
	80.8 100
	80.9 100
	81 100
	81.1 100
	81.2 100
	81.3 100
	81.4 100
	81.5 100
	81.6 100
	81.7 100
	81.8 100
	81.9 100
	82 100
	82.1 100
	82.2 100
	82.3 100
	82.4 100
	82.5 100
	82.6 100
	82.7 100
	82.8 100
	82.9 100
	83 100
	83.1 100
	83.2 100
	83.3 100
	83.4 100
	83.5 100
	83.6 100
	83.7 100
	83.8 100
	83.9 100
	84 100
	84.1 100
	84.2 100
	84.3 100
	84.4 100
	84.5 100
	84.6 100
	84.7 100
	84.8 100
	84.9 100
	85 100
	85.1 100
	85.2 100
	85.3 100
	85.4 100
	85.5 100
	85.6 100
	85.7 100
	85.8 100
	85.9 100
	86 100
	86.1 100
	86.2 100
	86.3 100
	86.4 100
	86.5 100
	86.6 100
	86.7 100
	86.8 100
	86.9 100
	87 100
	87.1 100
	87.2 100
	87.3 100
	87.4 100
	87.5 100
	87.6 100
	87.7 100
	87.8 100
	87.9 100
	88 100
	88.1 100
	88.2 100
	88.3 100
	88.4 100
	88.5 100
	88.6 100
	88.7 100
	88.8 100
	88.9 100
	89 100
	89.1 100
	89.2 100
	89.3 100
	89.4 100
	89.5 100
	89.6 100
	89.7 100
	89.8 100
	89.9 100
	90 100
	90.1 100
	90.2 100
	90.3 100
	90.4 100
	90.5 100
	90.6 100
	90.7 100
	90.8 100
	90.9 100
	91 100
	91.1 100
	91.2 100
	91.3 100
	91.4 100
	91.5 100
	91.6 100
	91.7 100
	91.8 100
	91.9 100
	92 100
	92.1 100
	92.2 100
	92.3 100
	92.4 100
	92.5 100
	92.6 100
	92.7 100
	92.8 100
	92.9 100
	93 100
	93.1 100
	93.2 100
	93.3 100
	93.4 100
	93.5 100
	93.6 100
	93.7 100
	93.8 100
	93.9 100
	94 100
	94.1 100
	94.2 100
	94.3 100
	94.4 100
	94.5 100
	94.6 100
	94.7 100
	94.8 100
	94.9 100
	95 100
	95.1 100
	95.2 100
	95.3 100
	95.4 100
	95.5 100
	95.6 100
	95.7 100
	95.8 100
	95.9 100
	96 100
	96.1 100
	96.2 100
	96.3 100
	96.4 100
	96.5 100
	96.6 100
	96.7 100
	96.8 100
	96.9 100
	97 100
	97.1 100
	97.2 100
	97.3 100
	97.4 100
	97.5 100
	97.6 100
	97.7 100
	97.8 100
	97.9 100
	98 100
	98.1 100
	98.2 100
	98.3 100
	98.4 100
	98.5 100
	98.6 100
	98.7 100
	98.8 100
	98.9 100
	99 100
	99.1 100
	99.2 100
	99.3 100
	99.4 100
	99.5 100
	99.6 100
	99.7 100
	99.8 100
	99.9 100
	100 100
};
\addplot [line width=3pt, darkorange25512714]
table {%
0 0
0.1 0
0.2 0
0.3 0
0.4 0
0.5 0
0.6 0
0.7 0
0.8 0
0.9 0
1 0
1.1 0
1.2 0
1.3 0
1.4 0
1.5 0
1.6 0
1.7 0
1.8 0
1.9 0
2 0
2.1 0
2.2 0
2.3 0
2.4 0
2.5 0
2.6 0
2.7 0
2.8 0
2.9 0
3 0
3.1 0
3.2 0
3.3 0
3.4 0
3.5 0
3.6 0
3.7 0
3.8 0
3.9 0
4 0
4.1 0
4.2 0
4.3 0
4.4 0
4.5 0
4.6 0
4.7 0
4.8 0
4.9 0
5 0
5.1 0
5.2 0
5.3 0
5.4 0
5.5 0
5.6 0
5.7 0
5.8 0
5.9 0
6 0
6.1 0
6.2 0
6.3 0
6.4 0
6.5 0
6.6 0
6.7 0
6.8 0
6.9 0
7 0
7.1 0
7.2 0
7.3 0
7.4 0
7.5 0
7.6 0
7.7 0
7.8 0
7.9 0
8 0
8.1 0
8.2 0
8.3 0
8.4 0
8.5 0
8.6 0
8.7 0
8.8 0
8.9 0
9 0
9.1 0
9.2 0
9.3 0
9.4 0
9.5 0
9.6 0
9.7 0
9.8 0
9.9 0
10 0
10.1 0
10.2 0
10.3 0
10.4 0
10.5 0
10.6 0
10.7 0
10.8 0
10.9 0
11 0
11.1 0
11.2 0
11.3 0
11.4 0
11.5 0
11.6 0
11.7 0
11.8 0
11.9 0
12 0
12.1 0
12.2 0
12.3 0
12.4 0
12.5 0
12.6 0
12.7 0
12.8 0
12.9 0
13 0
13.1 0
13.2 0
13.3 0
13.4 0
13.5 0
13.6 0
13.7 0
13.8 0
13.9 0
14 0
14.1 0
14.2 0
14.3 0
14.4 0
14.5 0
14.6 0
14.7 0
14.8 0
14.9 0
15 0
15.1 0
15.2 0
15.3 0
15.4 0
15.5 0
15.6 0
15.7 0
15.8 0
15.9 0
16 0
16.1 0
16.2 0
16.3 0
16.4 0
16.5 0
16.6 0
16.7 0
16.8 0
16.9 0
17 0
17.1 0
17.2 0
17.3 0
17.4 0
17.5 0
17.6 0
17.7 0
17.8 0
17.9 0
18 0
18.1 0
18.2 0
18.3 0
18.4 0
18.5 0
18.6 0
18.7 0
18.8 0
18.9 0
19 0
19.1 0
19.2 0
19.3 0
19.4 0
19.5 0
19.6 0
19.7 0
19.8 0
19.9 0
20 0
20.1 0
20.2 0
20.3 0
20.4 0
20.5 0
20.6 0
20.7 0
20.8 0
20.9 0
21 0
21.1 0
21.2 0
21.3 0
21.4 0
21.5 0
21.6 0
21.7 0
21.8 0
21.9 0
22 0
22.1 0
22.2 0
22.3 0
22.4 0
22.5 0
22.6 0
22.7 0
22.8 0
22.9 0
23 0
23.1 0
23.2 0
23.3 0
23.4 0
23.5 0
23.6 0
23.7 0
23.8 0
23.9 0
24 0
24.1 0
24.2 0
24.3 0
24.4 0
24.5 0
24.6 0
24.7 0
24.8 0
24.9 0
25 0
25.1 0
25.2 0
25.3 0
25.4 0
25.5 0
25.6 0
25.7 0
25.8 0
25.9 0
26 0
26.1 0
26.2 0
26.3 0
26.4 0
26.5 0
26.6 0
26.7 0
26.8 0
26.9 0
27 0
27.1 0
27.2 0
27.3 0
27.4 0
27.5 0
27.6 0
27.7 0
27.8 0
27.9 0
28 0
28.1 0
28.2 0
28.3 0
28.4 0
28.5 0
28.6 0
28.7 0
28.8 0
28.9 0
29 0
29.1 0
29.2 0
29.3 0
29.4 0
29.5 0
29.6 0
29.7 0
29.8 0
29.9 0
30 0
30.1 0
30.2 0
30.3 0
30.4 0
30.5 0
30.6 0
30.7 0
30.8 0
30.9 0
31 0
31.1 0
31.2 0
31.3 0
31.4 0
31.5 0
31.6 0
31.7 0
31.8 0
31.9 0
32 0
32.1 0
32.2 0
32.3 0
32.4 0
32.5 0
32.6 0
32.7 0
32.8 0
32.9 0
33 0
33.1 0
33.2 0
33.3 0
33.4 0
33.5 0
33.6 0
33.7 0
33.8 0
33.9 0
34 0
34.1 0
34.2 0
34.3 0
34.4 0
34.5 0
34.6 0
34.7 0
34.8 0
34.9 0
35 0
35.1 0
35.2 0
35.3 0
35.4 0
35.5 0
35.6 0
35.7 0
35.8 0
35.9 0
36 0
36.1 0
36.2 0
36.3 0
36.4 0
36.5 0
36.6 0
36.7 0
36.8 0
36.9 0
37 0
37.1 100
37.2 100
37.3 100
37.4 100
37.5 100
37.6 100
37.7 100
37.8 100
37.9 100
38 100
38.1 100
38.2 100
38.3 100
38.4 100
38.5 100
38.6 100
38.7 100
38.8 100
38.9 100
39 100
39.1 100
39.2 100
39.3 100
39.4 100
39.5 100
39.6 100
39.7 100
39.8 100
39.9 100
40 100
40.1 100
40.2 100
40.3 100
40.4 100
40.5 100
40.6 100
40.7 100
40.8 100
40.9 100
41 100
41.1 100
41.2 100
41.3 100
41.4 100
41.5 100
41.6 100
41.7 100
41.8 100
41.9 100
42 100
42.1 100
42.2 100
42.3 100
42.4 100
42.5 100
42.6 100
42.7 100
42.8 100
42.9 100
43 100
43.1 100
43.2 100
43.3 100
43.4 100
43.5 100
43.6 100
43.7 100
43.8 100
43.9 100
44 100
44.1 100
44.2 100
44.3 100
44.4 100
44.5 100
44.6 100
44.7 100
44.8 100
44.9 100
45 100
45.1 100
45.2 100
45.3 100
45.4 100
45.5 100
45.6 100
45.7 100
45.8 100
45.9 100
46 100
46.1 100
46.2 100
46.3 100
46.4 100
46.5 100
46.6 100
46.7 100
46.8 100
46.9 100
47 100
47.1 100
47.2 100
47.3 100
47.4 100
47.5 100
47.6 100
47.7 100
47.8 100
47.9 100
48 100
48.1 100
48.2 100
48.3 100
48.4 100
48.5 100
48.6 100
48.7 100
48.8 100
48.9 100
49 100
49.1 100
49.2 100
49.3 100
49.4 0
49.5 0
49.6 0
49.7 0
49.8 0
49.9 0
50 0
50.1 0
50.2 0
50.3 0
50.4 0
50.5 0
50.6 0
50.7 0
50.8 0
50.9 0
51 0
51.1 0
51.2 0
51.3 0
51.4 0
51.5 0
51.6 0
51.7 0
51.8 0
51.9 0
52 0
52.1 0
52.2 0
52.3 0
52.4 0
52.5 0
52.6 0
52.7 0
52.8 0
52.9 0
53 0
53.1 0
53.2 0
53.3 0
53.4 0
53.5 0
53.6 0
53.7 0
53.8 0
53.9 0
54 0
54.1 0
54.2 0
54.3 0
54.4 0
54.5 0
54.6 0
54.7 0
54.8 0
54.9 0
55 0
55.1 0
55.2 0
55.3 0
55.4 0
55.5 0
55.6 0
55.7 0
55.8 0
55.9 0
56 0
56.1 0
56.2 0
56.3 0
56.4 0
56.5 0
56.6 0
56.7 0
56.8 0
56.9 0
57 0
57.1 0
57.2 0
57.3 0
57.4 0
57.5 0
57.6 0
57.7 0
57.8 0
57.9 0
58 0
58.1 0
58.2 0
58.3 0
58.4 0
58.5 0
58.6 0
58.7 0
58.8 0
58.9 0
59 0
59.1 0
59.2 0
59.3 0
59.4 0
59.5 0
59.6 0
59.7 0
59.8 0
59.9 0
60 0
60.1 0
60.2 0
60.3 0
60.4 0
60.5 0
60.6 0
60.7 0
60.8 0
60.9 0
61 0
61.1 0
61.2 0
61.3 0
61.4 0
61.5 0
61.6 0
61.7 0
61.8 0
61.9 0
62 0
62.1 0
62.2 0
62.3 0
62.4 0
62.5 0
62.6 0
62.7 0
62.8 0
62.9 0
63 0
63.1 0
63.2 0
63.3 0
63.4 0
63.5 0
63.6 0
63.7 0
63.8 0
63.9 0
64 0
64.1 0
64.2 0
64.3 0
64.4 0
64.5 0
64.6 0
64.7 0
64.8 0
64.9 0
65 0
65.1 0
65.2 0
65.3 0
65.4 0
65.5 0
65.6 0
65.7 0
65.8 0
65.9 0
66 0
66.1 0
66.2 0
66.3 0
66.4 0
66.5 0
66.6 0
66.7 0
66.8 0
66.9 0
67 0
67.1 0
67.2 0
67.3 0
67.4 0
67.5 0
67.6 0
67.7 0
67.8 0
67.9 0
68 0
68.1 0
68.2 0
68.3 0
68.4 0
68.5 0
68.6 0
68.7 0
68.8 0
68.9 0
69 0
69.1 0
69.2 0
69.3 0
69.4 0
69.5 0
69.6 0
69.7 0
69.8 0
69.9 0
70 0
70.1 0
70.2 0
70.3 0
70.4 0
70.5 0
70.6 0
70.7 0
70.8 0
70.9 0
71 0
71.1 0
71.2 0
71.3 0
71.4 0
71.5 0
71.6 0
71.7 0
71.8 0
71.9 0
72 0
72.1 0
72.2 0
72.3 0
72.4 0
72.5 0
72.6 0
72.7 0
72.8 0
72.9 0
73 0
73.1 0
73.2 0
73.3 0
73.4 0
73.5 0
73.6 0
73.7 0
73.8 0
73.9 0
74 0
74.1 0
74.2 0
74.3 0
74.4 0
74.5 0
74.6 0
74.7 0
74.8 0
74.9 0
75 0
75.1 0
75.2 0
75.3 0
75.4 0
75.5 0
75.6 0
75.7 0
75.8 0
75.9 0
76 0
76.1 0
76.2 0
76.3 0
76.4 0
76.5 0
76.6 0
76.7 0
76.8 0
76.9 0
77 0
77.1 0
77.2 0
77.3 0
77.4 0
77.5 0
77.6 0
77.7 0
77.8 0
77.9 0
78 0
78.1 0
78.2 0
78.3 0
78.4 0
78.5 0
78.6 0
78.7 0
78.8 0
78.9 0
79 0
79.1 0
79.2 0
79.3 0
79.4 0
79.5 0
79.6 0
79.7 0
79.8 0
79.9 0
80 0
80.1 0
80.2 0
80.3 0
80.4 0
80.5 0
80.6 0
80.7 0
80.8 0
80.9 0
81 0
81.1 0
81.2 0
81.3 0
81.4 0
81.5 0
81.6 0
81.7 0
81.8 0
81.9 0
82 0
82.1 0
82.2 0
82.3 0
82.4 0
82.5 0
82.6 0
82.7 0
82.8 0
82.9 0
83 0
83.1 0
83.2 0
83.3 0
83.4 0
83.5 0
83.6 0
83.7 0
83.8 0
83.9 0
84 0
84.1 0
84.2 0
84.3 0
84.4 0
84.5 0
84.6 0
84.7 0
84.8 0
84.9 100
85 100
85.1 100
85.2 100
85.3 100
85.4 100
85.5 100
85.6 100
85.7 100
85.8 100
85.9 100
86 100
86.1 100
86.2 100
86.3 100
86.4 100
86.5 100
86.6 100
86.7 100
86.8 100
86.9 100
87 100
87.1 100
87.2 100
87.3 100
87.4 100
87.5 100
87.6 100
87.7 100
87.8 100
87.9 100
88 100
88.1 100
88.2 100
88.3 100
88.4 100
88.5 100
88.6 100
88.7 100
88.8 100
88.9 100
89 100
89.1 100
89.2 100
89.3 100
89.4 100
89.5 100
89.6 100
89.7 100
89.8 100
89.9 100
90 100
90.1 100
90.2 100
90.3 100
90.4 100
90.5 100
90.6 100
90.7 100
90.8 100
90.9 100
91 100
91.1 100
91.2 100
91.3 100
91.4 100
91.5 100
91.6 100
91.7 100
91.8 100
91.9 100
92 100
92.1 100
92.2 100
92.3 100
92.4 100
92.5 100
92.6 100
92.7 100
92.8 100
92.9 100
93 100
93.1 100
93.2 100
93.3 100
93.4 100
93.5 100
93.6 100
93.7 100
93.8 100
93.9 100
94 100
94.1 100
94.2 100
94.3 100
94.4 100
94.5 100
94.6 100
94.7 100
94.8 100
94.9 100
95 100
95.1 100
95.2 100
95.3 100
95.4 100
95.5 100
95.6 100
95.7 100
95.8 100
95.9 100
96 100
96.1 100
96.2 100
96.3 100
96.4 100
96.5 100
96.6 100
96.7 100
96.8 100
96.9 100
97 100
97.1 100
97.2 100
97.3 100
97.4 100
97.5 100
97.6 100
97.7 100
97.8 100
97.9 100
98 100
98.1 100
98.2 100
98.3 100
98.4 100
98.5 100
98.6 100
98.7 100
98.8 100
98.9 100
99 100
99.1 100
99.2 100
99.3 100
99.4 100
99.5 100
99.6 100
99.7 100
99.8 100
99.9 100
100 100
};
\end{axis}

\end{tikzpicture}

%% file: plots_MES/scatterplot_MESCwiel_adding_single_changing_costs.tex
\begin{tikzpicture}

\definecolor{darkgray176}{RGB}{176,176,176}
\definecolor{steelblue31119180}{RGB}{25,95,144}

\begin{axis}[
tick align=outside,
tick pos=left,
x grid style={darkgray176},
xlabel={$\singletonadd$},
xmin=0, xmax=1,
xtick style={color=black},
y grid style={darkgray176},
ylabel={$\costred$},
ymin=0, ymax=1,
ytick style={color=black},every tick label/.append style={font=\LARGE}, 
label style={font=\LARGE},
grid=both,
grid style={line width=.1pt, draw=gray!10}
]
\addplot [draw=black, fill=black, mark=*, only marks,mark size=1.5pt]
table{%
x  y
0.961538461538462 0.84
0.628930817610063 0.56
0.961538461538461 0.9
0.478468899521531 0.44
0.537634408602151 0.49
0.546448087431694 0.47
0.245700245700246 0.16
0.423728813559322 0.12
0.334448160535117 0.14
0.367647058823529 0.11
0.609756097560976 0.28
0.440528634361233 0.17
0.917431192660551 0.81
0.746268656716418 0.37
0.344827586206897 0.09
0.478468899521531 0.3
0.78125 0.71
0.714285714285714 0.18
0.892857142857143 0.54
0.628930817610063 0.33
0.645161290322581 0.52
0.423728813559322 0.16
0.478468899521531 0.25
0.478468899521531 0.18
0.602409638554217 0.29
0.355871886120996 0.14
0.251256281407035 0.07
0.568181818181818 0.24
0.636942675159236 0.34
0.564971751412429 0.29
0.892857142857143 0.84
};
\addplot [draw=steelblue31119180, fill=steelblue31119180, mark=*, only marks,mark size=1.8pt]
table{%
x  y
0.793650793650794 0.76
0.970873786407767 0.8
0.869565217391304 0.29
};
\fill[steelblue31119180!20!white] (axis cs:0.793650793650794,0.76) circle (0.3cm);
\fill[steelblue31119180!20!white] (axis cs:0.970873786407767,0.8) circle (0.3cm);
\fill[steelblue31119180!20!white] (axis cs:0.869565217391304,0.29) circle (0.3cm);
\draw (axis cs:0.961538461538462,0.84) node[
  scale=0.53,
  anchor=base west,
  text=black,
  rotate=0.0
]{\LARGE 8};
\draw (axis cs:0.628930817610063,0.56) node[
  scale=0.53,
  anchor=base west,
  text=black,
  rotate=0.0
]{\LARGE 13};
\draw (axis cs:0.949538461538461,0.91) node[
  scale=0.53,
  anchor=base west,
  text=black,
  rotate=0.0
]{\LARGE 16};
\draw (axis cs:0.793650793650794,0.76) node[
  scale=0.7,
  anchor=base west,
  text=steelblue31119180,
  font=\bfseries,
  rotate=0.0
]{\LARGE 18};
\draw (axis cs:0.940873786407767,0.74) node[
  scale=0.7,
  anchor=base west,
  text=steelblue31119180,
  font=\bfseries,
  rotate=0.0
]{\LARGE 21};
\draw (axis cs:0.478468899521531,0.44) node[
  scale=0.53,
  anchor=base west,
  text=black,
  rotate=0.0
]{\LARGE 27};
\draw (axis cs:0.535634408602151,0.50) node[
  scale=0.53,
  anchor=base west,
  text=black,
  rotate=0.0
]{\LARGE 30};
\draw (axis cs:0.546448087431694,0.44) node[
  scale=0.53,
  anchor=base west,
  text=black,
  rotate=0.0
]{\LARGE 31};
\draw (axis cs:0.245700245700246,0.16) node[
  scale=0.53,
  anchor=base west,
  text=black,
  rotate=0.0
]{\LARGE 38};
\draw (axis cs:0.423728813559322,0.12) node[
  scale=0.53,
  anchor=base west,
  text=black,
  rotate=0.0
]{\LARGE 44};
\draw (axis cs:0.3,0.16) node[
  scale=0.53,
  anchor=base west,
  text=black,
  rotate=0.0
]{\LARGE 47};
\draw (axis cs:0.367647058823529,0.1) node[
  scale=0.53,
  anchor=base west,
  text=black,
  rotate=0.0
]{\LARGE 48};
\draw (axis cs:0.599756097560976,0.237) node[
  scale=0.53,
  anchor=base west,
  text=black,
  rotate=0.0
]{\LARGE 51};
\draw (axis cs:0.431528634361233,0.19) node[
  scale=0.53,
  anchor=base west,
  text=black,
  rotate=0.0
]{\LARGE 52};
\draw (axis cs:0.875431192660551,0.755) node[
  scale=0.53,
  anchor=base west,
  text=black,
  rotate=0.0
]{\LARGE 54};
\draw (axis cs:0.746268656716418,0.37) node[
  scale=0.53,
  anchor=base west,
  text=black,
  rotate=0.0
]{\LARGE 55};
\draw (axis cs:0.344827586206897,0.05) node[
  scale=0.53,
  anchor=base west,
  text=black,
  rotate=0.0
]{\LARGE 59};
\draw (axis cs:0.478468899521531,0.31) node[
  scale=0.53,
  anchor=base west,
  text=black,
  rotate=0.0
]{\LARGE 63};
\draw (axis cs:0.78425,0.675) node[
  scale=0.53,
  anchor=base west,
  text=black,
  rotate=0.0
]{\LARGE 64};
\draw (axis cs:0.714285714285714,0.18) node[
  scale=0.53,
  anchor=base west,
  text=black,
  rotate=0.0
]{\LARGE 65};
\draw (axis cs:0.892857142857143,0.54) node[
  scale=0.53,
  anchor=base west,
  text=black,
  rotate=0.0
]{\LARGE 66};
\draw (axis cs:0.869565217391304,0.29) node[
  scale=0.7,
  anchor=base west,
  text=steelblue31119180,
  font=\bfseries,
  rotate=0.0
]{\LARGE 67};
\draw (axis cs:0.578930817610063,0.34) node[
  scale=0.53,
  anchor=base west,
  text=black,
  rotate=0.0
]{\LARGE 68};
\draw (axis cs:0.645161290322581,0.51) node[
  scale=0.53,
  anchor=base west,
  text=black,
  rotate=0.0
]{\LARGE 72};
\draw (axis cs:0.388728813559322,0.18) node[
  scale=0.53,
  anchor=base west,
  text=black,
  rotate=0.0
]{\LARGE 78};
\draw (axis cs:0.478468899521531,0.25) node[
  scale=0.53,
  anchor=base west,
  text=black,
  rotate=0.0
]{\LARGE 79};
\draw (axis cs:0.478468899521531,0.17) node[
  scale=0.53,
  anchor=base west,
  text=black,
  rotate=0.0
]{\LARGE 80};
\draw (axis cs:0.608409638554217,0.285) node[
  scale=0.53,
  anchor=base west,
  text=black,
  rotate=0.0
]{\LARGE 81};
\draw (axis cs:0.348871886120996,0.15) node[
  scale=0.53,
  anchor=base west,
  text=black,
  rotate=0.0
]{\LARGE 82};
\draw (axis cs:0.251256281407035,0.07) node[
  scale=0.53,
  anchor=base west,
  text=black,
  rotate=0.0
]{\LARGE 83};
\draw (axis cs:0.568181818181818,0.2) node[
  scale=0.53,
  anchor=base west,
  text=black,
  rotate=0.0
]{\LARGE 84};
\draw (axis cs:0.636942675159236,0.33) node[
  scale=0.53,
  anchor=base west,
  text=black,
  rotate=0.0
]{\LARGE 85};
\draw (axis cs:0.53471751412429,0.305) node[
  scale=0.53,
  anchor=base west,
  text=black,
  rotate=0.0
]{\LARGE 86};
\draw (axis cs:0.892857142857143,0.84) node[
  scale=0.53,
  anchor=base west,
  text=black,
  rotate=0.0
]{\LARGE 87};
\end{axis}
\end{tikzpicture}